\documentclass{IEEEtran}

\usepackage{amsmath}
\usepackage{amssymb}
\usepackage{mathdots}
\usepackage{quantum}
\usepackage[hypertexnames=false,pdfpagelabels]{hyperref}
\usepackage[amsmath,thmmarks,hyperref]{ntheorem}
\usepackage[capitalise,poorman]{cleveref}

\crefname{equations}{Eqs.}{Eqs.}
\Crefname{equations}{Equations}{Equations}
\creflabelformat{equations}{\textup{(#2#1#3)}}
\crefrangelabelformat{equations}{\textup{(#3#1#4)}--\textup{(#5#2#6)}}

\newtheorem{theorem}{Theorem}
\newtheorem{lemma}[theorem]{Lemma}
\newtheorem{corollary}[theorem]{Corollary}
\newtheorem{proposition}[theorem]{Proposition}
\newtheorem{definition}[theorem]{Definition}

\DeclareMathOperator{\supp}{supp}
\DeclareMathAlphabet{\mathitbf}{OML}{cmm}{b}{it}
\let\vec\relax
\DeclareMathOperator{\vec}{vec}
\newcommand{\cB}{\mathcal{B}}
\newcommand{\chan}{\mathcal{E}}
\newcommand{\gr}{\mathrm{Gr}}
\newcommand{\keyword}[1]{\emph{#1}}
\newcommand{\flip}{\mathbb{F}}
\DeclareMathOperator{\swap}{SWAP}
\newcommand{\mat}{\mathbb{M}}
\renewcommand{\emptyset}{\varnothing}

\begin{document}

\title{Superactivation of the Asymptotic Zero-Error
  Classical Capacity of a Quantum Channel}

\author{Toby S. Cubitt, Jianxin Chen and Aram W. Harrow%
  \thanks{T.~S.~Cubitt was supported by a Leverhulme early-career
    fellowship. J.~Chen was supported by The China Scholarship Council.
    A.~W.~Harrow and T.~S.~Cubitt acknowledge support through the
    integrated EC project ``QAP'' (contract no.~IST-2005-15848), and
    A.~W.~Harrow was also funded by the U.K. EPSRC grant ``QIP IRC.''}
  \thanks{This work was carried out when T.~S.~Cubitt and A.~W.~Harrow were
    with the Department of Mathematics, University of Bristol, United
    Kingdom, and J.~Chen was with the State Key Laboratory of Intelligent
    Technology and Systems, Department of Computer Science and
    Technology, Tsinghua University.}%
  \thanks{T.~S.~Cubitt is with the Departamento de An\'alisis
    Matem\'atico, Universidad Complutense de Madrid, Plaza de Ciencias~3,
    Ciudad Universitaria, 28040~Madrid, Spain (email:
    tcubitt@mat.ucm.es).}
  \thanks{A.~W.~Harrow is with the Dept.\ of Comp.\ Sci.\ \& Eng., U.\ of
    Washington, Seattle, WA 98195, USA (email: aram@cs.washington.edu).}
  \thanks{J.~Chen is with the Department of Mathematics \& Statistics,
    University of Guelph, Guelph, Ontario, Canada, and the Institute for
    Quantum Computing, University of Waterloo, Waterloo, Ontario, Canada
    (email: chenkenshin@gmail.com)}}

\maketitle

\begin{abstract}
  The zero-error classical capacity of a quantum channel is the
  asymptotic rate at which it can be used to send classical bits
  perfectly, so that they can be decoded with zero probability of
  error. We show that there exist pairs of quantum channels, neither of
  which individually have any zero-error capacity whatsoever (even if
  arbitrarily many uses of the channels are available), but such that
  access to even a single copy of both channels allows classical
  information to be sent perfectly reliably. In other words, we prove
  that the zero-error classical capacity can be superactivated. This
  result is the first example of superactivation of a \emph{classical}
  capacity of a quantum channel.
\end{abstract}


\section{Introduction}\label{sec:intro}
Shannon's information theory has been highly successful at describing
classical information transmission, but only in the last couple of
decades or so has there been a major effort to extend it to quantum
channels, and even quantum information, that we must contend with in the
real world. A major strength of Shannon's work is that the calculation of
asymptotic capacities, although potentially requiring optimisations over
unbounded numbers of channel uses, typically reduces to a simple, and
often convex, optimisation problem over a single use of a channel (a
\keyword{single-letter formula}). Moreover, many of these capacities are
\keyword{additive}, meaning that access to two channels together allows
one to send information at a rate equal to the sum of the channels'
individual capacities. These two properties---additivity, and the
reduction from the asymptotic capacity to a single-letter formula---are
both crucial to the elegance of Shannon's theory. The latter allows us to
compute capacities, and the former tells us that this single number
completely characterises the channel's usefulness for classical
information transmission.

Accordingly, in quantum information theory the most important questions
in extending Shannon's techniques concern \keyword{additivity} (whether
the capacity of two channels together is ever greater than the sum of
their individual capacities) and \keyword{regularisation} (whether the
asymptotic capacity of a channel can be reduced to optimising an entropic
quantity over a single use of a channel). The classical and quantum
capacities of a quantum channel can be expressed in terms of the
regularised asymptotic limits of the \keyword{Holevo
  capacity}~\cite{Holevo,Schumacher+Westmoreland} and \keyword{coherent
  information}~\cite{Devetak,Shor,Lloyd}, respectively. There was an
early hope that the quantum capacity of a quantum channel might be
expressed in terms of the maximum coherent information from a
\emph{single} use of the channel, and that the classical capacity could
be similarly expressed in terms of the Holevo capacity. However, this
hope proved to be unfounded. The maximum coherent information and Holevo
capacity turn out not to equal the channel capacities. This was proved
over a decade ago for the quantum capacity~\cite{DSS97}, and only in the
last year for the classical capacity~\cite{Hastings} (the culmination of
a series of similar results~\cite{Andreas+Patrick,rank_additivity} for
minimum output R\'enyi entropies). This implies that entangling inputs
across different channel uses is in general necessary for optimal quantum
channel coding. It also tells us that if single-letter formulae exist for
the quantum and classical capacities, they will not equal the maximum
coherent information or the Holevo capacity.

However, these results tell us only that regularisation is necessary for
our existing formula, not that the quantum channel capacities are
necessarily non-additive. The first demonstration of non-additivity was
given recently by Smith and Yard~\cite{graeme+jon}, who showed that the
quantum capacity is super-additive. Indeed, their result proved that
additivity is violated in the strongest possible sense: they exhibited
two quantum channels which, individually, have zero quantum
capacity. Yet, combine the two, and the joint channel has non-zero
capacity. In other words, not only is the quantum capacity non-additive,
there even exist channels that are completely useless for transmitting
quantum information, but which \emph{can} transmit quantum information
when used together. The term ``\keyword{superactivation}'' was coined in
Ref.~\cite{SST00} to describe this phenomenon, since the two channels
somehow ``activate'' each other's hidden ability to transmit quantum
information. More recent work has established the nonadditivity of the
private classical capacity~\cite{private1,private2}. On the other hand,
additivity of the classical capacity of a quantum channel remains an open
question.

The Shannon capacity, and the classical and quantum capacities mentioned
so far, all measure the capacity for transmitting information with an
error probability that can be made arbitrarily small, in the limit of
arbitrarily many uses of the channel. Right from the early days of his
development of classical information theory, Shannon also considered the
\keyword{zero-error capacity}: the capacity of a channel to transmit
information perfectly, with zero probability of
error~\cite{Shannon_zero-error}. The zero-error capacity is important for
applications in which no error can be tolerated, and also, and perhaps
more importantly, when only a limited number of uses of the channel are
available, so that the low error probability for the Shannon capacity can
not be achieved.

Even in the case of classical channels, the zero-error capacity turns out
to be mathematically very different to the standard Shannon capacity. For
example, it is known to be non-additive. (See e.g.\
Ref.~\cite{zero-error_review} for a review of zero-error information
theory.) However, it is not difficult to see that there can be no
superactivation of the zero-error capacity of a classical channel. The
main result of our paper shows that for \emph{quantum} channels this is
no longer true; the zero-error classical capacity of a quantum channel
\emph{can} be superactivated:
\begin{theorem}\label{thm:main}
  Let $d_A=16, d_E = 4(2d_A-1)= 124$ and $d_B = d_Ad_E = 1984$.  Then
  there exist channels $\chan_1,\chan_2$ such that:
  \begin{itemize}
  \item Each channel $\chan_{1,2}$ maps $\CC^{d_A}$ to $\CC^{d_B}$ and
    has $d_E$ Kraus operators.
  \item Each channel $\chan_{1,2}$ has no zero-error capacity.
  \item The joint channel $\chan_1\otimes\chan_2$ \emph{does} have
    non-zero zero-error capacity.
  \end{itemize}
\end{theorem}

In other words, there exist pairs of quantum channels that individually
cannot be used for perfect transmission of \emph{any} classical
information at all, even if infinitely many uses of the channel are
available. Yet, when the two channels are combined, even a \emph{single}
use of the each of the two channels allows perfect, error-free
transmission of classical information. To our knowledge, this is the
first example of superactivation of any kind of classical capacity of
standard quantum channels.

Naturally, similar results also hold for larger-dimensional input and
output spaces. Increasing the output dimension is trivial, since the
channels do not need to make use of the entire output space. To increase
the input dimension without changing the results of the theorem, we
define channels $\hat{\chan}_{1,2}$ that act as follows: on the first 16
dimensions of the input $\hat{\chan}_{1,2}$ match the behaviour of
$\chan_{1,2}$, and the remaining dimensions are mapped to a maximally
mixed state on the output.

The definition of zero-error capacity is easily extended to the quantum
setting~\cite{MA05}. Beigi and Shor investigated the computational
complexity of computing the zero-error capacity of quantum
channels~\cite{BS07}, showing that it is in general difficult to
compute. Most notably, and one of the main inspirations for this work,
Duan and Shi~\cite{Duan+Shi} proved a ``one-shot'' result in the case of
multi-sender/multi-receiver quantum channels, when the senders and
receivers are restricted to local operations and classical communication
(LOCC). They exhibited examples of such channels for which a \emph{single
  use} has no zero-error classical capacity but two uses do have non-zero
zero-error capacity.

Duan and Shi's work hints at superactivation of the asymptotic capacity
for standard quantum channels. Indeed, it raises two tantalising
questions. Are these remarkable properties of the zero-error capacity
inherent to communication over quantum channels, or do they arise from
the LOCC constraints in the multi-sender/multi-receiver setting, which
are crucial for their proofs? Furthermore, are their results an artifact
of the one-shot case, that would disappear in the asymptotic setting?
Both questions are compellingly answered by our work. This paper is also
in some sense a sequel to our earlier work in
Ref.~\cite{rank_additivity}, which demonstrated non-multiplicativity of
the one-shot minimum output rank of a quantum channel, and its extension
to the asymptotic case in Ref.~\cite{Jianxin_strong}. (The relation
between this problem and the superactivation phenomenon will be explained
in \cref{sec:one-shot}.)

The paper is organised as follows. \Cref{sec:preliminaries} introduces
the necessary notation and concepts, and \cref{sec:conjugate-divisible}
proves some basic mathematical properties of composite quantum maps that
play a key role later. In \cref{sec:one-shot}, we prove a one-shot
version of the main result. This is presented in some detail because,
firstly, the main result builds directly on techniques used to prove the
one-shot case and, secondly, in the one-shot case we are able to give
explicit examples which may give some insight into the main result. In
\cref{sec:asymptotic}, we draw on techniques from algebraic geometry to
prove our main result: superactivation of the asymptotic zero-error
classical capacity of quantum channels. Finally, we conclude in
\cref{sec:conclusions} with a discussion of the results and their
implications.

\section{Preliminaries}\label{sec:preliminaries}

\subsection{Quantum channels}
The complex conjugate of $x$ will be denoted $\bar{x}$. The adjoint
$\chan^*$ of a map $\chan$ on the space $\mathcal{B}(\HS)$ of bounded
operators on $\HS$ is the dual with respect to the Hilbert-Schmidt
inner-product, i.e.\ the unique map defined by
\begin{equation}
  \tr[A^\dagger\,\chan(B)] = \tr[\,\chan^*(A)^\dagger\,B].
\end{equation}
Alternatively, any map $\chan$ on $\cB(\HS)$ can be written as $\chan(X)
= \sum_k A_k X B_k$. In this representation, $\chan^*(X) = \sum_k A_k^\dg
X B_k^\dg$.

A map $\chan$ on $\mathcal{B}(\HS)$ is \keyword{completely positive} (CP)
if it not only maps all positive operators to positive operators, but
also preserves positivity when applied to a subsystem of some larger
system. In this case, it can be written in the Kraus form $\chan(X) =
\sum_k E_k X E_k^\dag$, and $\chan^*(X) = \sum_k E_k^\dag X E_k$. A CP
map is completely positive and \keyword{trace-preserving} (CPT) if it in
addition preserves the trace of operators. (CPT maps in quantum mechanics
play exactly the analogous role to communication channels in classical
information theory, and we will use the terms \keyword{quantum channel}
and CPT map synonymously.)

The ``flip'' operation on a bipartite state is the composition of the
swap operation, which interchanges the two parties, and complex
conjugation:
\begin{equation}
  \flip(\ket[AB]{\psi}) = \swap(\ket[AB]{\bar{\psi}}).
\end{equation}
(Note that the complex conjugation means the flip operation is
basis-dependent; the computational product basis should be assumed when
no basis is stated explicitly.) Thus, with complex-conjugation defined in
the computational basis,
\begin{equation}
  \flip\Bigl(\sum_{ij}c_{ij}\ket[A]{i}\ket[B]{j}\Bigr)
  = \sum_{ij}\bar{c}_{ij}\ket[A]{j}\ket[B]{i}.
\end{equation}
The definition of the flip operation extends to operators as $\flip(M) =
\swap\cdot\bar{M}\cdot\swap$.

\begin{definition}
  We say that a bipartite state or operator is
  \keyword{conjugate-symmetric} in a given basis if it is invariant under
  the flip operation, and similarly for a subspace invariant under the
  same operation.
\end{definition}

There is a straightforward isomorphism between (unnormalised) states
$\ket[AB]{\psi}$ in a bipartite space $\CC^{d_A}\otimes\CC^{d_B}$ and
$d_A\times d_B$ matrices $M$: writing $\ket{\psi}$ in a product basis, we
have
\begin{equation}
  \ket[AB]{\psi} = \sum_{ij}M_{ij}\ket{i}\ket{j}.
\end{equation}
We will write $\mat(\ket{\psi})$ when we wish to denote the coefficient
matrix $M$ corresponding to the state $\ket{\psi}$. Similarly, we denote
by $\mat(S)$ the matrix subspace isomorphic in this way to a subspace
$S\subseteq\HS_A\otimes\HS_B$. In terms of these coefficient matrices, a
conjugate-symmetric state is one for which $\mat(\ket{\psi})$ is
Hermitian, and a subspace is conjugate-symmetric iff the corresponding
matrix space is spanned by a basis of Hermitian matrices. Note that the
\keyword{Schmidt-rank} of the state $\ket{\psi}$ is exactly the linear
rank of $\mat(\ket{\psi})$.
\begin{definition}
  We say that a bipartite state $\ket[AB]{\psi}$ is
  \keyword{positive-semidefinite} in a given product basis if
  $\mat(\ket{\psi})$ is a positive-semidefinite matrix. (Note that this
  includes the statement that $\mat(\ket{\psi})$ is Hermitian.)
  Similarly, a positive-semidefinite subspace $S_{AB}$ is one that admits
  a \emph{basis} whose elements are all positive-semidefinite.
\end{definition}
Note that it is obviously \emph{not} the case that all elements of a
positive-semidefinite subspace $\mat(S_{AB})$ need themselves be
positive-semidefinite, just that there exists \emph{some} set of
positive-semidefinite elements that span the space. Indeed, the existence
of even a single positive-\emph{definite} element is sufficient, as we
can then make any basis positive-semidefinite by adding sufficient weight
of this positive-definite element to every basis state. A
positive-semidefinite subspace is necessarily conjugate-symmetric, by
definition.

\begin{definition}
  We say a map $\mathcal{N}$ is \keyword{conjugate-divisible} if it can
  be decomposed as $\mathcal{N} = \chan^*\circ\chan$ for some CPT map
  $\chan$.
\end{definition}
(Note that a necessary condition for conjugate-divisibility of
$\mathcal{N}$ is that the matrix representation of $\mathcal{N}$ as a
superoperator be a positive-semidefinite matrix. This follows from the
fact that, if $E$ is the matrix representation $\chan$, then the matrix
representation of $\chan^*\circ\chan$ is $E^\dg E$, which is necessarily
positive-semidefinite. However, this is not sufficient, since
conjugate-divisibility carries the additional non-trivial requirement
that $\chan$ be CPT.)

It will frequently be convenient to work with the Choi-Jamio\l{}kowski
representation of a map. Recall that the Choi-Jamio\l{}kowski matrix
associated with a map $\chan$ is the matrix $\sigma_{AB} =
\mathcal{I}_A\otimes\chan_B(\omega_{AB})$ obtained by applying the map to
one half of the (unnormalised) full Schmidt-rank state $\ket{\omega} =
\sum_i\lambda_i\ket[A]{\varphi_i}\ket[B]{\chi_i}$. This isomorphism holds
regardless of whether $\chan$ is a CPT map or not; iff $\chan$ is CP(T),
then $\sigma/\tr\proj{\omega}$ is a (trace 1) positive operator. (The
standard Choi-Jamio\l{}kowski matrix $\tilde{\sigma}_{AB}$ is obtained by
setting $\ket{\omega} = \sum_i\ket{i}\ket{i}$, but the isomorphism holds
more generally.) Introducing the unitary basis change $U\ket{\varphi_i} =
\ket{\chi_i}$, We can recover the action of the map $\chan$ from the
matrix $\sigma_{AB}$ via
\begin{equation}
  \chan(\rho)
  = \tr_A\left[U\sigma_A^{-1/2}\sigma_{AB}\,\sigma_A^{-1/2}U^\dagger\cdot
               \rho^T\otimes\id\right],
\end{equation}
where $\sigma_A = \tr_B[\sigma_{AB}]$. For the standard
Choi-Jamio\l{}kowski matrix $\tilde{\sigma}_{AB}$, this simplifies to
$\chan(\rho) = \tr_A[\,\tilde{\sigma}_{AB}\cdot\rho^T\otimes\id]$, and
the non-standard Choi-Jamio\l{}kowski matrix $\sigma_{AB}$ is related to
the standard one by rotating and rescaling the $A$ subsystem:
\begin{equation}
  \tilde{\sigma}_{AB}
  = U\sigma_A^{-1/2}\sigma_{AB}\,\sigma_A^{-1/2}U^\dagger.
\end{equation}

\subsection{{Basic algebraic geometry concepts}}
The proof of our main theorem requires certain mathematical tools from
basic algebraic geometry. For convenience of the reader, we recall
some definitions and results in algebraic geometry. For more details, we
refer to~\cite{Hartshorne77, Shafarevich94}.

Let $\mathbb{A}^n$ be an \keyword{affine} $n$-space, the set of all
$n$-tuples of complex numbers. Denote $\mathbb{C}[x_1,x_2,\cdots,x_n]$
as the polynomial ring in $n$ variables. A subset of $\mathbb{A}^n$ is
an \keyword{algebraic set} or \keyword{algebraic variety} if it
consists of the common zeros of a finite set of polynomials
$f_1,f_2,\cdots,f_r$ with $f_i\in \mathbb{C}[x_1,x_2,\cdots,x_n]$ for
all $1\leq i\leq r$. Such an algebraic set is usually denoted by
$Z(f_1,f_2,\cdots,f_r)$. By taking the open subsets to be the
complements of algebraic sets, we can define a topology on
$\mathbb{A}^n$, called the Zariski topology. The
\keyword{Zariski-closed} sets are then precisely the algebraic sets.
(Note that in some references the term \keyword{algebraic variety} is
reserved for varieties that are irreducible, in the sense that they
cannot be expressed as the union of two proper algebraic sets.)

We define \keyword{projective $n$-space}, denoted by $\mathbb{P}^n$, to
be the set of equivalence classes of $(n+1)-$tuples $(a_0,\cdots,a_n)$ of
complex numbers, not all zero, under the equivalence relation given by
$(a_0,\cdots,a_n)\sim(\lambda a_0,\cdots,\lambda a_n)$ for all $\lambda
\in \mathbb{C}$, $\lambda\neq 0$.

Similarly, a subset $Y$ of $\mathbb{P}^n$ is an \keyword{algebraic set}
or \keyword{projective variety} if it consists of the common zeros of a
finite set of homogeneous polynomials $f_1,f_2,\cdots,f_r$ with $f_i\in
\mathbb{C}[x_0,x_1,\cdots,x_n]$ for $1\leq i\leq r$.

\section{Conjugate-divisible maps}\label{sec:conjugate-divisible}
The composite map $\chan^*\circ\chan$ will turn out to play a key role in
studying the zero-error capacity of the channel $\chan$. So we will first
need to establish some basic properties of such conjugate-divisible
maps. The main goal is a complete characterisation of their
Choi-Jamio\l{}kowski matrices.

\begin{lemma}\label{lem:Choi_conjugate}
  If $\rho_{AB}$ is the (standard) Choi-Jamio\l{}kowski matrix for a
  channel $\chan$, then the (standard) Choi-Jamio\l{}kowski matrix of
  $\chan^*$ is given by $\flip(\rho_{AB}) = \bar{\rho}_{BA}$.
\end{lemma}
\begin{IEEEproof}
  We have
  \begin{subequations}
  \begin{align}
    \tr\left[\chan^*(\psi)^\dagger\,\varphi\right]
    &=\tr\left[\psi^\dagger\,\chan(\varphi)\right]\\
    &=\tr\left[
       \psi^\dagger\,\tr_A\left(\rho_{AB}\cdot\varphi^T\otimes\id\right)
     \right]\\
    &=\tr\left[
        \id\otimes\psi^\dagger\cdot\rho_{AB}^{T_A}\cdot\varphi\otimes\id
      \right]\\
    &=\tr\left[
        \tr_B\left(\id\otimes\psi\cdot\bar{\rho}_{AB}^{T_B}\right)^\dagger
        \cdot\varphi
      \right]\\
    &=\tr\left[
       \tr_B\left(\bar{\rho}_{BA}^{T_B}\cdot\psi\otimes\id\right)^\dagger
       \cdot\varphi
     \right]\\
    &=\tr\left[
        \tr_B\left(\flip(\rho_{AB})\cdot\psi^T\otimes\id\right)^\dagger
        \cdot\varphi
      \right],
  \end{align}
  \end{subequations}
  from which we identify the Choi-Jamio\l{}kowski matrix for $\chan^*$ to
  be as claimed.
\end{IEEEproof}

\begin{lemma}\label{lem:Choi_composite}
  If $\rho_{AB}$ is the (standard) Choi-Jamio\l{}kowski matrix for a
  channel $\chan$, then the (standard) Choi-Jamio\l{}kowski matrix of
  $\mathcal{N} = \chan^*\circ\chan$ is given by
  \begin{equation}
    \sigma_{AA'}
    =\tr_B\left[
        \rho_{AB}^{\phantom{T_B}}\otimes\id_{A'}\cdot
        \id_A\otimes\bar{\rho}_{BA'}^{T_B}
      \right].
  \end{equation}
\end{lemma}
\begin{IEEEproof}
  We have
  \begin{subequations}
  \begin{align}
    \mathcal{N}(\psi)
    &=\tr_B\left[
        \bar{\rho}_{BA'}^{\phantom{T_B}}\cdot
        \tr_A\left(\rho_{AB}\cdot\psi^T\otimes\id_B\right)^T
        \otimes\id_{A'}
      \right]\\
    &=\tr_B\left[
        \tr_A\left(\rho_{AB}\cdot\psi^T\otimes\id_B\right)
        \otimes\id_{A'}\cdot
        \bar{\rho}_{BA'}^{T_B}
      \right]\\
    &=\tr_A\left[
      \tr_B\left(
        \rho_{AB}^{\phantom{T_B}}\otimes\id_{A'}\cdot
        \id_A\otimes\bar{\rho}_{BA'}^{T_B}
      \right)
      \cdot\psi^T\otimes\id_{A'}
    \right],
  \end{align}
  \end{subequations}
  from which we identify the Choi-Jamio\l{}kowski matrix of $\mathcal{N}$
  to be as claimed.
\end{IEEEproof}

The following extension to non-standard Choi-Jamio\l{}kowski matrices
follows immediately.
\begin{corollary}\label{cor:Choi_non-standard}
  If $\rho_{AB}$ is a non-standard Choi-Jamio\l{}kowski matrix for a
  channel $\chan$, related to the standard Choi-Jamio\l{}kowski matrix by
  \begin{equation}
    \tilde{\rho}_{AB}
    = U\rho_A^{-1/2}\rho_{AB}\,\rho_A^{-1/2}U^\dagger,
  \end{equation}
  then
  \begin{equation}
    \sigma_{AA'}
    =\tr_B\left[
        \rho_{AB}^{\phantom{T_B}}\otimes\id_{A'}\cdot
        \id_A\otimes\bar{\rho}_{BA'}^{T_B}
      \right]
  \end{equation}
  can be viewed as a non-standard Choi-Jamio\l{}kowski matrix for
  $\mathcal{N}=\chan^*\circ\chan$ by identifying it with the standard
  Choi-Jamio\l{}kowski matrix $\tilde{\sigma}_{AA'}$ for $\mathcal{N}$ in
  the following way:
  \begin{equation}
    \tilde{\sigma}_{AA'}
    = U\sigma_A^{-1/2}\otimes \bar{U}\bar{\sigma}_{A'}^{-1/2}\cdot
      \sigma_{AA'}\cdot
      \sigma_A^{-1/2}U^\dagger\otimes\bar{\sigma}_{A'}^{-1/2}\bar{U}^\dagger.
  \end{equation}
\end{corollary}

With these basic properties in hand, we are now in a position to prove a
necessary condition for a matrix to be the Choi-Jamio\l{}kowski matrix of
some conjugate-divisible map.
\begin{proposition}\label{prop:necessary}
  The support of the Choi-Jamio\l{}kowski matrix of a conjugate-divisible
  map is positive-semidefinite (hence conjugate-symmetric).
\end{proposition}
\begin{IEEEproof}
  To establish conjugate-symmetry, let $\mathcal{N}=\chan^*\circ\chan$ be
  conjugate-divisible, where $\chan:A\to B$ is CPT, and denote the
  (standard) Choi-Jamio\l{}kowski matrix of $\chan$ by $\rho_{AB}$. By
  \cref{lem:Choi_composite}, the Choi-Jamio\l{}kowski matrix of
  $\mathcal{N}$ is given by
  \begin{equation}
    \sigma_{AA'}
    =\tr_B\left[
        \rho_{AB}^{\phantom{T_B}}\otimes\id_{A'}\cdot
        \id_A\otimes\bar{\rho}_{BA'}^{T_B}
      \right].
  \end{equation}
  Hence
  \begin{subequations}
  \begin{align}
    \flip(\sigma_{AA'})
    &=\flip\left(
        \tr_B\left[
          \rho_{AB}^{\phantom{T_B}}\otimes\id_{A'}\cdot
          \id_A\otimes\bar{\rho}_{BA'}^{T_B}
        \right]
      \right)\\
    &=\tr_B\left[
      \id_A\otimes\bar{\rho}_{BA'}\cdot\rho_{AB}^{T_B}\otimes\id_{A'}
    \right]\\
    &=\tr_B\left[
      \rho_{AB}^{\phantom{AB}}\otimes\id_{A'}\cdot
      \id_A\otimes\bar{\rho}_{BA'}^{T_B}
    \right]\\
    &=\sigma_{AA'}.
  \end{align}
  \end{subequations}
  Since $\sigma_{AA'}$ is conjugate-symmetric, so is its support (i.e.\
  the support is invariant as a subspace under the action of $\flip$).

  To establish positive-semidefiniteness, first write the eigenvectors
  $\ket{\varphi_k}$ of $\rho_{AB}$ in a product basis:
  \begin{equation}
    \rho_{AB} = \sum_k\proj{\varphi_k},\qquad
    \ket[AB]{\varphi_k} = \sum_i\ket[A]{\psi_i^k}\ket[B]{i},
  \end{equation}
  where the eigenvalues and coefficients have been absorbed into the
  unnormalised states $\ket[AB]{\varphi_k}$ and $\ket[A]{\psi_i^k}$ (note
  also that $\ket[A]{\psi_i^k}$ are not necessarily orthogonal). Then
  \begin{subequations}
  \begin{align}
    \sigma_{AA'}
    &=\tr_B\left[
      \rho_{AB}^{\phantom{AB}}\otimes\id_{A'}\cdot
      \id_A\otimes\bar{\rho}_{BA'}^{T_B}
    \right]\\[0.5em]
    &\begin{aligned}
      =\tr_B\Biggl[
        &\sum_{ijk}\ket{\psi_i^k}\ket{i}\bra{\psi_j^k}\bra{j}
          \otimes\id_{A'}
        \cdot\\[-1em]
        &\mspace{80mu}
         \id_A\otimes
         \sum_{lmn}\ket{n}\ket{\bar{\psi}_m^l}\bra{m}\bra{\bar{\psi}_n^l}
      \Biggr]
    \end{aligned}\\[0.5em]
    &=\sum_{ijkl}\ket{\psi_i^k}\ket{\bar{\psi}_i^l}
        \bra{\psi_j^k}\bra{\bar{\psi}_j^l}\\
    &=\sum_{kl}
      \Bigl(\sum_i\ket{\psi_i^k}\ket{\bar{\psi}_i^l}\Bigr)
      \Bigl(\sum_j\bra{\psi_j^k}\bra{\bar{\psi}_j^l}\Bigr),
  \end{align}
  \end{subequations}
  from which we see that
  \begin{equation}
    S_{AA'}
    = \supp(\sigma_{AA'})
    = \vspan\Bigl\{
        \sum_i\ket{\psi_i^k}\ket{\bar{\psi}_i^l}
      \Bigr\}_{k,l}.
  \end{equation}
  Now, as matrices
  \begin{equation}
    \mat\Bigl(\sum_i\ket{\psi_i^k}\ket{\bar{\psi}_i^l}\Bigr)
    = \sum_i\ket{\psi_i^k}\bra{\psi_i^l},
  \end{equation}
  which are supported on $\vspan\{\ket{\psi_i^k}\}$. In particular, the
  matrix subspace $\mat(S_{AA'})$ contains
  \begin{equation}
    \mat\Bigl(\sum_{ik}\ket{\psi_i^k}\ket{\bar{\psi}_i^k}\Bigr)
    = \sum_{ik}\ket{\psi_i^k}\bra{\psi_i^k}
  \end{equation}
  which has full support on the subspace $\vspan\{\ket{\psi_i^k}\}$ and,
  being a sum of (unnormalised) projectors, has positive eigenvalues on
  that subspace. Thus we can choose as a basis for $\mat(S_{AA'})$ the
  set of matrices
  \begin{equation}
    \begin{split}
      &\biggl\{
        \sum_j\left(
          \ket{\psi_j^k}\bra{\psi_j^l} + \ket{\psi_j^l}\bra{\psi_j^k}
        \right)
        + c\sum_{j,k}\ket{\psi_j^k}\bra{\psi_j^k},\\
        &\quad
        \sum_j i\left(
          \ket{\psi_j^k}\bra{\psi_j^l} - \ket{\psi_j^l}\bra{\psi_j^k}
        \right)
        + c\sum_{j,k}\ket{\psi_j^k}\bra{\psi_j^k}
      \biggr\}_{k,l}
    \end{split}\raisetag{4em}
  \end{equation}
  which are all Hermitian and, for sufficiently large $c$,
  positive-semidefinite.
\end{IEEEproof}

We now show that the necessary conditions of \cref{prop:necessary} are
also sufficient.
\begin{proposition}\label{prop:sufficient}
  For any conjugate-symmetric, positive-semidefinite subspace $S_{AA'}$
  which has full support on the first subsystem (i.e.\
  $\supp(\tr_{A'}[S_{AA'}]) = \HS_A$), we can construct a (in general
  non-standard) Choi-Jamio\l{}kowski matrix $\sigma_{AA'}$ of a
  conjugate-divisible map such that $\supp(\sigma_{AA'}) = S_{AA'}$. The
  corresponding channel $\chan$ has input dimension $d_A$, rank $d_E=\dim
  S_{AA'}$ and output dimension $d_B=d_Ad_E$.
\end{proposition}
(Here, the notation $\supp(\tr_{A'}[S_{AA'}])$ is shorthand for
$\bigcup_{\ket{\psi}\in S_{AA'}}\supp(\tr_{A'} \proj{\psi})$. The
condition on the support is necessary for a matrix to be any kind of
Choi-Jamio\l{}kowski matrix, simply by definition.)
\begin{IEEEproof}
  Since $S_{AA'}$ is positive-semidefinite, we can choose a Hermitian
  basis $\{M_k\}$ for $\mat(S_{AA'})$ such that $M_k \geq 0$. Writing
  $M_k$ in its spectral decomposition,
  \begin{equation}
    M_k = \sum_i\proj{\psi_i^k},
  \end{equation}
  where we have absorbed the (positive) eigenvalues into the unnormalised
  eigenstates $\ket{\psi_i^k}$, we have
  \begin{equation}
    S_{AA'} = \vspan\Bigl\{\sum_i\ket{\psi_i^k}\ket{\bar{\psi}_i^k}\Bigr\}_k
  \end{equation}
  and $\HS_A = \vspan\{\ket{\psi_i^k}\}$.

  Now consider the operator
  \begin{equation}
    \rho_{AB}
    = \sum_{ijk}\ket[A]{\psi_i^k}\ket[B]{k,i}\bra[A]{\psi_j^k}\bra[B]{k,j}.
  \end{equation}
  This is Hermitian, positive-semidefinite, and $\tr_B[\rho_{AB}]$ is
  full rank on $\HS_A$, so (up to normalisation) $\rho_{AB}$ is a
  (non-standard) Choi-Jamio\l{}kowski matrix corresponding to some CPT
  map $\chan$. Observe also that the rank and local dimensions of
  $\rho_{AB}$ are as claimed in the statement of the proposition. By
  \cref{cor:Choi_non-standard},
  \begin{subequations}
  \begin{align}
    \sigma_{AA'}
    &= \tr_B\left[
      \rho_{AB}^{\phantom{AB}}\otimes\id_A\cdot
      \id_A\otimes\bar{\rho}_{BA'}^{T_B}
    \right]\\[0.5em]
    &\begin{aligned}
      =\tr_B\Biggl[
        &\sum_{ijk}\ket{\psi_i^k}\ket{k,i}
          \bra{\psi_j^k}\bra{k,j}\otimes\id_{A'}\,\cdot\\[-1em]
        &\mspace{50mu}
         \id_A\otimes\sum_{lmn}\ket{l,n}\ket{\psi_m^l}\bra{l,m}\bra{\psi_n^l}
      \Biggr]
    \end{aligned}\\[0.5em]
    &=\sum_k\Bigl(\sum_i\ket{\psi_i^k}\ket{\bar{\psi}_i^k}\Bigr)
      \Bigl(\sum_j\bra{\psi_j^k}\bra{\bar{\psi}_j^k}\Bigr)
  \end{align}
  \end{subequations}
  is a (non-standard) Choi-Jamio\l{}\-kow\-ski matrix for the
  conjugate-divisible channel $\chan^*\circ\chan$. Clearly, the support
  of this operator is $S_{AA'}$, so it fulfils the requirements of the
  proposition.
\end{IEEEproof}

\Cref{prop:necessary,prop:sufficient} together imply the following key
theorem, giving a complete characterisation of the
Choi-Jamio\l{}\-kow\-ski matrices of conjugate-divisible maps.
\begin{theorem}\label{thm:conjugate_divisible}
  Given a subspace $S_{AA'}$ such that $\supp(\tr_{A'}[S_{AA'}]) =
  \HS_A$, there exists a conjugate-divisible map with (in general
  non-standard) Choi-Ja\-mio\l{}\-kow\-ski matrix $\sigma_{AA'}$ such
  that $\supp(\sigma_{AA'}) = S_{AA'}$ iff $S_{AA'}$ is
  positive-semidefinite (hence also conjugate-symmetric).
\end{theorem}

\section{Superactivation of the one-shot zero-error capacity}
\label{sec:one-shot}
The \keyword{zero-error classical capacity} of a quantum channel is the
capacity to transmit classical information with zero probability of error
(as opposed to a vanishing error probability, as in the usual Shannon
capacity; for brevity, we will drop the ``classical'' nomenclature from
now on, and call this simply the \keyword{zero-error capacity}). The
\keyword{one-shot zero-error capacity} is the amount of (classical)
information that can be transmitted with zero probability of error by a
\emph{single} use of the channel (as opposed to the asymptotic rate per
use of the channel in the limit of infinitely many uses of the channel).
Our aim in this section is to show that there exist two quantum channels,
which individually have zero one-shot zero-error capacity, but whose
joint channel \emph{does} have a non-zero zero-error capacity. (In
\cref{sec:asymptotic}, we will extend this result to the asymptotic
capacity.)

A channel $\chan$ has non-zero (one-shot) zero-error capacity if there
exist two different input states whose outputs are perfectly
distinguishable. In other words, the one-shot zero-error capacity is
non-zero iff
\begin{equation}
   \exists \ket{\psi},\ket{\varphi}\in\HS_A:
   \tr[\chan(\psi)^\dagger\chan(\varphi)] = 0.
\end{equation}
Note that
\begin{equation}
  \tr[\chan(\psi)^\dagger\chan(\varphi)]
   = \tr[\psi\cdot\chan^*\bigl(\chan(\varphi)\bigr)]
   = \tr[\psi\cdot\chan^*\circ\chan(\varphi)].
\end{equation}
Conversely, a channel has zero one-shot zero-error capacity iff
\begin{equation}
   \forall \ket{\psi},\ket{\varphi}\in\HS_A:
   \tr[\psi\cdot\chan^*\circ\chan(\varphi)] \neq 0.
\end{equation}
Thus we seek two channels, $\chan_1$ and $\chan_2$, such that
\begin{subequations}
\begin{gather}
   \forall \ket{\psi},\ket{\varphi}\in\HS_A:
   \tr[\psi\cdot\chan_{1,2}^*\circ\chan_{1,2}(\varphi)] \neq 0,
   \label{eq:individual}\\
   \exists \ket{\psi},\ket{\varphi}\in\HS_A^{\otimes 2}:
   \tr[\psi\cdot
         (\chan_1^*\circ\chan_1)\otimes(\chan_2^*\circ\chan_2)(\varphi)]
   = 0. \label{eq:joint}
\end{gather}
\end{subequations}
For the composite maps $\mathcal{N}_{1,2} =
\chan_{1,2}^*\circ\chan_{1,2}$ these are precisely the conditions
established in Ref.~\cite{rank_additivity} for $\mathcal{N}_{1,2}$ to
violate multiplicativity of the minimum output rank! The composite map
$\mathcal{N} = \chan^*\circ\chan$ need not be CPT even if $\chan$ is, but
this does not substantially affect the methods developed in
Ref.~\cite{rank_additivity}, which we will reuse here.

To establish necessary and sufficient conditions for the individual maps
to satisfy \cref{eq:individual}, we follow exactly the same arguments as
in Ref.~\cite{rank_additivity}. Let $\sigma_{1,2}$ denote
Choi-Jamio\l{}kowski matrices corresponding to the conjugate-divisible
maps $\mathcal{N}_{1,2}$. Then, from \cref{eq:individual}, we have
\begin{equation}
  \begin{split}
    &\forall \ket{\psi},\ket{\varphi}\in\HS_A:\\
    &\tr\left[
       \psi_{A'}\cdot\tr_A(\sigma_{1,2}\cdot\varphi_A^T\otimes\id_{A'})
     \right]
    =\tr\left[\sigma_{1,2}\cdot\varphi_A^T\otimes\psi_{A'}\right]
    \neq 0.
  \end{split}
\end{equation}
Note that this holds even if $\sigma_{1,2}$ are non-standard
Choi-Jamio\l{}kowski matrices, since using \cref{cor:Choi_non-standard}
any rescaling can be absorbed into $\varphi$ and $\psi$:
\begin{subequations}
\begin{align}
  &\begin{aligned}
    \tr\Bigl[
      &U\sigma_A^{-1/2}\otimes\bar{U}\bar{\sigma}_A^{-1/2}\cdot
      \sigma_{AA'}\cdot\\
      &\mspace{100mu}
       \sigma_A^{-1/2}U^\dagger\otimes\bar{\sigma}_A^{-1/2}\bar{U}^\dagger
      \cdot\varphi_A\otimes\psi_{A'}
    \Bigr]
  \end{aligned}\\
  &\begin{aligned}
    =\tr\Bigl[
      \sigma_{AA'}\cdot
      \Bigl(
        &\sigma_A^{-1/2}U^\dagger\varphi_A U\sigma_A^{-1/2}\otimes\\
        &\mspace{100mu}
         \bar{\sigma}_A^{-1/2}\bar{U}^\dagger\psi_{A'} \bar{U}\bar{\sigma}_A^{-1/2}
      \Bigr)
    \Bigr]\\
  \end{aligned}\\
  &=\tr\left[\sigma_{AA'}\cdot\varphi'_A\otimes\psi'_{A'}\right].
\end{align}
\end{subequations}
Therefore, if $S_{1,2} = \supp(\sigma_{1,2})$ denote the supports of the
Choi-Jamio\l{}kowski matrices, it is necessary and sufficient to require
that their orthogonal complements contain no product states:
\begin{equation}\label{eq:no_product_states}
  \nexists\ket{\psi},\ket{\varphi}\in\HS_{A}\; : \;
  \ket{\psi}\otimes\ket{\varphi} \in S_{1,2}^\perp.
\end{equation}

To derive sufficient conditions for the joint map to satisfy
\cref{eq:joint}, we slightly generalise the argument of
Ref.~\cite{rank_additivity}. First, fix both states $\ket{\psi}$,
$\ket{\varphi}$ in \cref{eq:joint} to be maximally entangled: $\ket{\psi}
= U_{A_1}\otimes V_{A_2}\ket{\omega}$, $\ket{\varphi} = W_{A'_1}\otimes
X_{A'_2}\ket{\omega}$, where $\ket{\omega} = \sum_i\ket{i,i}$ and
$U,V,W,X$ are unitary. Then
\begin{subequations}
\begin{align}
  0 &=\tr\left[\psi_{A'_1A'_2}\cdot
      \mathcal{N}_1\otimes\mathcal{N}_2(\varphi_{A_1A_2})\right]\\[0.5em]
  &=\tr\left[\psi_{A'_1A'_2}\cdot
      \tr_{A_1A_2}\left[
        \sigma_1\otimes\sigma_2\cdot\varphi_{A_1A_2}^T\otimes\id_{A'_1A'_2}
      \right]
    \right]\\[0.5em]
  &=\tr\left[
      \sigma_1\otimes\sigma_2\cdot
      \varphi_{A_1A_2}^T\otimes\psi_{A'_1A'_2}
    \right]\\[0.5em]
  &\begin{aligned}
   =\tr\bigl[
      &\sigma_1\otimes\sigma_2\cdot
      (\bar{U}\otimes\bar{V}\;\omega_{A_1A_2}^T\; U^T\otimes V^T)
      \otimes\\
      &\mspace{158mu}
       (W\otimes X\;\omega_{A'_1A'_2}\; W^\dagger\otimes X^\dagger)
    \bigr]
  \end{aligned}\raisetag{2.8em}\\[0.5em]
  &=\tr\left[
     (\bar{U}\otimes W\,\sigma_1\,U^T\otimes W^\dagger)^T\cdot
     (\bar{V}\otimes X\,\sigma_2\,V^T\otimes X^\dagger)
    \right]\\
  &=\tr\left[
      \sigma_1^T\cdot
      (U'\otimes V'\,\sigma_2\,{U'}^\dagger\otimes{V'}^\dagger)
    \right].\label{eq:orthogonal}
\end{align}
\end{subequations}
Again, this remains true if $\sigma_{1,2}$ are non-standard
Choi-Jamio\l{}kowski matrices, since we can absorb any rescaling into our
choice of $\ket{\psi}$ and $\ket{\varphi}$. Writing
$U_{1,2}\sigma_{A_{1,2}}^{-1/2} = R_{1,2}$ for brevity, we have
\begin{subequations}
\begin{align}
  &\tr\left[\psi_{A'_1A'_2}\cdot
    \mathcal{N}_1\otimes\mathcal{N}_2(\varphi_{A_1A_2})\right]\\
  &\begin{aligned}
    =\tr\Bigl[
      &\left(R_1\otimes\bar{R}_1\otimes R_2\otimes\bar{R}_2\right)\cdot
       \sigma_{A_1A'_1}\otimes\sigma_{A_2A'_2}\;\cdot\\
      &\mspace{30mu}
       \left(R_1\otimes\bar{R}_1\otimes R_2\otimes\bar{R}_2\right)\cdot
       \varphi_{A_1A_2}^T\otimes\psi_{A'_1A'_2}
    \Bigr]
  \end{aligned}\raisetag{2em}\\
  &\begin{aligned}
    =&\tr\Bigl[
      \sigma_{A_1A'_1}\otimes\sigma_{A_2A'_2}\cdot
      (R_1\otimes\bar{R}_1\;\varphi_{A_1A_2}^TR_1\otimes\bar{R}_1)\\
     &\mspace{150mu}
      \otimes(R_2\otimes\bar{R}_2\;\psi_{A'_1A'_2}R_2\otimes\bar{R}_2)
    \Bigr]
  \end{aligned}\\
  &=\tr\left[\sigma_{A_1A'_1}\otimes\sigma_{A_2A'_2}\cdot
    {\varphi'}_{A_1A_2}^T\otimes\psi'_{A'_1A'_2}\right].
\end{align}
\end{subequations}
Therefore, in terms of the supports $S_{1,2}$ of the
Choi-Jamio\l{}\-kow\-ski matrices $\sigma_{1,2}$, \cref{eq:orthogonal}
implies that a sufficient condition for the maps to satisfy
\cref{eq:joint} is for the supports to be related by
\begin{equation}\label{eq:orthogonal_complement}
  S_2^T = U\otimes V\cdot S_1^\perp
\end{equation}
for some local unitaries $U,V$.

Of course, since $\mathcal{N}_{1,2}=\chan_{1,2}^*\circ\chan_{1,2}$ are
necessarily conjugate-divisible, \cref{thm:conjugate_divisible} also
applies, so $S_{1,2}$ must also be positive-semidefinite (hence
conjugate-symmetric). If we can find subspaces simultaneously satisfying
these conditions \emph{and}
\cref{eq:no_product_states,eq:orthogonal_complement}, then by
\cref{thm:conjugate_divisible} we can construct channels $\chan_{1,2}$
such that $\mathcal{N}_{1,2} = \chan_{1,2}^*\circ\chan_{1,2}$ satisfy
\cref{eq:individual,eq:joint}. (Note that w.l.o.g.\ we can neglect the
condition in \cref{thm:conjugate_divisible} that
$\supp(\tr_{A'}[S_{AA'}]) = \HS_A$, since if this is not the case we can
always shrink $\HS_A$ so that it does hold.) Noting that Schmidt-rank,
conjugate-symmetry and positive-semidefiniteness are preserved under the
transpose operation, we can for convenience redefine $S_2 =
\supp(\sigma_2^T)$ in \cref{eq:orthogonal_complement} (without changing
\cref{eq:no_product_states}) to save carrying the transpose around in the
notation.

These results are summarised in the following lemma:
\begin{lemma}\label{lem:S1_S2_conditions}
  If there exist subspaces $S_1,S_2 \subseteq \HS_A\otimes\HS_A$ and
  unitaries $U,V$ satisfying
  \begin{subequations}
  \begin{gather}
    \nexists\ket{\psi},\ket{\varphi}\in\HS_A:
      \ket{\psi}\otimes\ket{\varphi} \in S_{1,2}^\perp\;,
      \label{eq:unextendible}\\
    S_2 = U\otimes V\cdot S_1^\perp,
      \label{eq:conspiracy}\\
    \flip(S_{1,2}) = S_{1,2}\;,
      \label{eq:conjugate_symmetry}\\
    \exists \{M_i^{1,2} \geq 0\}: \mat(S_{1,2}) = \vspan\{M^{1,2}_i\},
      \label{eq:positivity}
  \end{gather}
  \end{subequations}
  then there exist channels $\chan_{1,2}$ which individually have zero
  one-shot zero-error capacity, but for which the joint channel
  $\chan_1\otimes\chan_2$ has non-zero zero-error capacity.
\end{lemma}
Although a positive-semidefinite subspace is necessarily
conjugate-symmetric, it will be convenient in what follows to treat
conjugate-symmetry separately from the positive-semidefinite requirement.
We therefore redundantly include the conjugate-symmetry requirement as
well as the positive-semidefinite requirement in the statement of this
and subsequent lemmas.

If a subspace is conjugate-symmetric, then so is its orthogonal
complement, so \cref{eq:conspiracy,eq:conjugate_symmetry} together imply
\begin{align}\label{eq:conj_sym_U}
  U\otimes V\cdot S_1 = S_2^\perp = \flip(S_2^\perp)
  = \flip(U\otimes V\cdot S_1).
\end{align}
Conversely, if \cref{eq:conj_sym_U} holds for conjugate-symmetric $S_1$,
then clearly \cref{eq:conjugate_symmetry} is satisfied. Thus, letting
$S_1 = S$, $S_2 = U\otimes V\cdot S^\perp$, and recalling that
Schmidt-rank is invariant under local-unitaries,
\cref{eq:unextendible,eq:conjugate_symmetry} can, respectively, be
re-expressed as:
\begin{gather}
  \nexists\ket{\psi},\ket{\varphi}\in\HS_A:
    \ket{\psi}\otimes\ket{\varphi} \in S \text{ or } S^\perp,
    \tag{\ref{eq:unextendible}'}\\
  \flip(S) = S \;\text{ and }\;
    \flip(U\otimes V\cdot S) = U\otimes V\cdot S.
    \tag{\ref{eq:conjugate_symmetry}'}
\end{gather}
We can therefore rewrite \cref{lem:S1_S2_conditions} in terms of a single
subspace $S$:
\begin{theorem}\label{thm:S_conditions}
  If there exists a subspace $S \subseteq \HS_A\otimes\HS_A$ and
  unitaries $U,V$ satisfying
  \begin{subequations}\label[equations]{eq:S_conditions}
  \begin{gather}
    \nexists\ket{\psi},\ket{\varphi}\in\HS_A:
      \ket{\psi}\otimes\ket{\varphi} \in S^\perp,
      \label{eq:S_unextendible}\\
    \nexists\ket{\psi},\ket{\varphi}\in\HS_A:
      \ket{\psi}\otimes \ket{\varphi} \in S,
      \label{eq:Sperp_unextendible}\\
    \flip(S) = S\,,\label{eq:S_conjugate_symmetry}\\
    \flip(U\otimes V\cdot S) = U\otimes V\cdot S,
      \label{eq:Sperp_conjugate_symmetry}\\
    \exists \{M_i \geq 0\}: \mat(S) = \vspan\{M_i\},
      \label{eq:S_positivity}\\
    \exists \{M_j \geq 0\}: \mat(U\otimes V\cdot S^\perp) = \vspan\{M_j\},
      \label{eq:Sperp_positivity}
  \end{gather}
  \end{subequations}
  then there exist channels $\chan_{1,2}$ which individually have zero
  one-shot zero-error capacity, but for which the joint channel
  $\chan_1\otimes\chan_2$ has non-zero zero-error capacity.
\end{theorem}

Our task, then, reduces to finding a subspace $S$ along with unitaries
$U,V$ which satisfy the conditions of \cref{thm:S_conditions}. (The first
two conditions are identical to those required in
Ref.~\cite{rank_additivity}. The remainder arise from the additional
conjugate-divisibility requirement, which rules out the explicit example
constructed in that paper.) Using the ideas of
Refs.~\cite{rank_additivity,Schmidt_rank_subspace}, it is not too hard to
find an explicit example of a subspace satisfying
\cref{thm:S_conditions}. For example, set
\begin{equation}
  U = \id,\quad
  V = \begin{pmatrix}&&&1\\ &&1\\ &1\\1\end{pmatrix} =: X,
\end{equation}
and choose the matrix subspace $\mat(S_1)$ to be spanned by
\begin{equation}
  \begin{aligned}
    &\begin{pmatrix}1\\ &1\\ &&1\\ &&&1\end{pmatrix},&&
     \begin{pmatrix}1\\ &i\\ &&-i\\ &&&-1\end{pmatrix},\\
    &\begin{pmatrix}1\\ &-i\\ &&i\\ &&&-1\end{pmatrix},&&
     \begin{pmatrix}
       1&&&1\\ &-1&-1\\ &-1&-1\\ 1&&&1
     \end{pmatrix},\\
    &\begin{pmatrix}
      0&-4&7\\ &&&\phantom{-}7\\ &&&-4\\ &&&\phantom{-}0
     \end{pmatrix},&&
     \begin{pmatrix}
       \phantom{-}0\\ -4&&&\\\phantom{-}7&&&\\ &7&-4&0
     \end{pmatrix},\\
    &\begin{pmatrix}
       0&-8&9\\ &&&-9\\ &&&\phantom{-}8\\ &&&\phantom{-}0
     \end{pmatrix},&&
     \begin{pmatrix}
       \phantom{-}0\\ -8&&&\\\phantom{-}9&&&\\ &-9&8&0
     \end{pmatrix}.
  \end{aligned}
\end{equation}
(The entries of the final four matrices are fairly arbitrary; they were
essentially chosen by picking two different sets of four integers at
random, and symmetrising.)

$\mat(S_1^\perp)$ is then spanned by
\begin{gather}
  \begin{aligned}
    &\begin{pmatrix}&&&1\\ &&1\\ &1\\1\end{pmatrix},&&
     \begin{pmatrix}&&&1\\ &&i\\ &-i\\-1\end{pmatrix},\\
    &\begin{pmatrix}&&&1\\ &&-i\\ &i\\-1\end{pmatrix},&&
     \begin{pmatrix}
       1&&&1\\ &\phantom{-}1&-1\\ &-1&\phantom{-}1\\ 1&&&1
     \end{pmatrix},\\
    &\begin{pmatrix}
       0&1&2\\ &&&-6\\ &&&-8\\ &&&\phantom{-}0
     \end{pmatrix},&&
     \begin{pmatrix}0\\1&&&\\2&&&\\ &-6&-8&0\end{pmatrix},\\
    &\begin{pmatrix}
       0&-8&-6\\ &&&2\\ &&&1\\ &&&\phantom{-}0
     \end{pmatrix},&&
     \begin{pmatrix}\phantom{-}0\\-8&&&\\-6&&&\\ &2&1&0\end{pmatrix}.
   \end{aligned}
\end{gather}
It is straightforward to verify that this choice of $S_1$ satisfies the
conjugate-symmetry conditions of
\cref{eq:S_conjugate_symmetry,eq:Sperp_conjugate_symmetry}. To see that
the positive-semi\-def\-in\-ite\-ness conditions of
\cref{eq:S_positivity,eq:Sperp_positivity} are satisfied, note that $S_1$
and $\id\otimes X\cdot S_1^\perp$ both contain the identity matrix, which
is positive and full rank. Thus we can construct a positive-semidefinite
basis by adding sufficient weight of the identity to the other basis
elements. Finally, the easiest way to prove that
\cref{eq:S_unextendible,eq:Sperp_unextendible} are satisfied is to use a
computer algebra package such as Mathematica, and apply the Groebner
basis algorithm. (Note that this provides a rigorous computer-aided
proof, not merely supporting numerical evidence.)

\section{Superactivation of the asymptotic zero-error capacity}
\label{sec:asymptotic}
We have proven in the previous section that the one-shot zero-error
capacity can be superactivated, which hints at an even more remarkable
possibility: can the \emph{asymptotic} capacity be superactivated?

The main challenge lies in showing that a channel has zero zero-error
capacity even in the asymptotic limit. This involves proving that
\emph{all} tensor powers of the channel have zero zero-error capacity.
From the arguments of \cref{sec:one-shot}, this implies that the
orthogonal complement of any tensor power of the support of its
Choi-Jamio\l{}kowski matrix should contain no product states. Thus, as in
\cref{sec:one-shot}, our task is to find a subspace that satisfies all
the conditions of \cref{eq:S_conditions}, but we strengthen
\cref{eq:S_unextendible,eq:Sperp_unextendible} to in addition require
that no tensor powers of the subspaces contain any product states. Given
such a subspace, we can construct a pair of channels in exactly the same
way as we did in \cref{sec:one-shot}, but thanks to these stronger
properties the individual channels will now have zero zero-error capacity
even in the asymptotic limit. This is summarised in the following
counterpart to \cref{thm:S_conditions}. (Once again, it is helpful for
later to redundantly retain the conjugate-symmetry requirement of
\cref{eq:S_conjugate_symmetry2,eq:Sperp_conjugate_symmetry2}, even though
this is already implied by the positive-semidefinite requirement of
\cref{eq:S_positivity2,eq:Sperp_positivity2}.)
\begin{theorem}\label{thm:S_strongly_unextendible_conditions}
  If there exists a subspace $S$ and unitaries $U,V$ satisfying
  \begin{subequations}
  \label[equations]{eq:S_strongly_unextendible_conditions}
  \begin{gather}
    \forall k,\nexists\ket{\psi},\ket{\varphi}\in\HS_A^{\otimes k}:
      \ket{\psi}\otimes\ket{\varphi} \in (S^{\otimes k})^\perp,
      \label{eq:S_strongly_unextendible}\\
    \forall k,\nexists\ket{\psi},\ket{\varphi}\in\HS_A^{\otimes k}:
      \ket{\psi}\otimes \ket{\varphi}
      \in \bigl((S^\perp)^{\otimes k}\bigr)^\perp,
      \label{eq:Sperp_strongly_unextendible}\\
    \flip(S) = S\,,\label{eq:S_conjugate_symmetry2}\\
    \flip(U\otimes V\cdot S) = U\otimes V\cdot S,
      \label{eq:Sperp_conjugate_symmetry2}\\
    \exists \{M_i \geq 0\}: \mat(S) = \vspan\{M_i\},
      \label{eq:S_positivity2}\\
    \exists \{M_j \geq 0\}: \mat(U\otimes V\cdot S^\perp) = \vspan\{M_j\},
      \label{eq:Sperp_positivity2}
  \end{gather}
  \end{subequations}
  then there exist channels $\chan_{1,2}$ which individually have no
  zero-error capacity, but whose joint channel $\chan_1\otimes\chan_2$
  \emph{does} have non-zero zero-error capacity.
\end{theorem}

Before proving that such a subspace exists, it is worth outlining the
general approach. We first adapt and extend the algebraic-geometry
arguments of Ref.~\cite{Jianxin_strong} to show that either almost all
subspaces satisfying
\cref{eq:S_conjugate_symmetry2,eq:Sperp_conjugate_symmetry2} also satisfy
\cref{eq:S_strongly_unextendible}, or none of them do. Then, we construct
a particular subspace that \emph{does} satisfy
\cref{eq:S_conjugate_symmetry2,eq:Sperp_conjugate_symmetry2,eq:S_strongly_unextendible}. Whilst
that particular subspace certainly does \emph{not} satisfy
\cref{eq:Sperp_strongly_unextendible}, the fact that it exists shows that
almost all subspaces satisfying
\cref{eq:S_conjugate_symmetry2,eq:Sperp_conjugate_symmetry2} must also
satisfy \cref{eq:S_strongly_unextendible}. And, by symmetry, this implies
that almost all of them also satisfy
\cref{eq:Sperp_strongly_unextendible}. Therefore, if we choose a subspace
satisfying \cref{eq:S_conjugate_symmetry2,eq:Sperp_conjugate_symmetry2}
at random, it will almost-surely satisfy
\cref{eq:S_strongly_unextendible,eq:Sperp_strongly_unextendible}. Finally,
we show that there is a non-zero probability that such a randomly chosen
subspace will also satisfy \cref{eq:S_positivity2,eq:Sperp_positivity2},
implying that a subspace satisfying all the conditions in
\cref{thm:S_strongly_unextendible_conditions} does exist.

\subsection{Strongly unextendible conjugate-symmetric subspaces are
  full measure}
We first require some terminology, notation and basic results relating to
the first two conditions,
\cref{eq:S_strongly_unextendible,eq:Sperp_strongly_unextendible}, of
\cref{thm:S_strongly_unextendible_conditions}.
\begin{definition}
  A subspace $S\subset\HS_A\otimes\HS_B$ is $k$-unextendible if
  $(S^{\otimes k})^\perp$ contains no product state in $\HS_{A^{\otimes
      k}}\otimes\HS_{B^{\otimes k}}$. A subspace is \keyword{strongly
    unextendible} if it is $k$-unextendible for all $k\geq 1$.
  Conversely, a subspace is \keyword{$k$-extendible} if it is not
  $k$-unextendible, and \keyword{extendible} if it is not strongly
  unextendible.
\end{definition}

$\gr_d(V)$ denotes the Grassmannian of a vector space $V$ (the set of all
$d$-dimensional subspaces of $V$). The sets of $k$-extendible,
extendible, and strongly unextendible subspaces of dimension $d$ will be
denoted, respectively,
\begin{align}
  E_d^k(\HS_A,\HS_B)
    &=\{S\in\gr_d(\HS_A\otimes\HS_B)
        \,|\, S \text{ is $k$-extendible}\},\\
  E_d(\HS_A,\HS_B)
    &=\{S\in\gr_d(\HS_A\otimes\HS_B)
        \,|\, S \text{ is extendible}\},
\end{align}
\begin{multline}
  U_d(\HS_A,\HS_B)= \\
    \{S\in\gr_d(\HS_A\otimes\HS_B)
        \,|\, S \text{ is strongly unextendible}\},
\end{multline}
so that
\begin{equation}
  U_d(\HS_A,\HS_B) = \Bigl(\bigcup_k E_d^k(\HS_A,\HS_B)\Bigr)^c,
\end{equation}
i.e.\ $U_d$ is the complement of the union over all $E_d^k$.

We start by proving that $E_d^k$ is an algebraic set:
\begin{lemma}\label{lem:Edk-closed}
  $E_d^k(\HS_A,\HS_B)$ is Zariski-closed in
  $\gr_d(\HS_A\otimes\HS_B) =
  \gr_d(\CC^{d_A}\otimes\CC^{d_B})$.
\end{lemma}

Before proving this lemma, we need some background about complete
varieties and proper morphisms. We will state here only the necessary
facts, without introducing formal mathematical definitions.

A continuous function between topological spaces is \keyword{proper} if
inverse images of compact subsets are compact. In algebraic geometric
settings, an analogue of a compact set is a \keyword{complete} variety.
For our purposes, it suffices to know that every projective variety is
complete, and a variety over $\mathbb{C}$ is complete if and only if it
is compact in the classical topology.

Similarly, a proper morphism between varieties is an analogue of a proper
map between \emph{classical} topological spaces. We will make key use of
some basic properties of proper morphisms. First, inverse images of
complete varieties are complete too. Second, the composition of two
proper morphisms is proper again. Thirdly, projective morphisms are
proper.

For those interested in formal definitions and more detailed properties,
we refer to \cite{Shafarevich94}. With these basic facts, we are now in a
position to prove our lemma.

\begin{IEEEproof}
  Define the following two maps:
  \begin{subequations}
  \begin{align}
    &\begin{aligned}
      \phi_1: &\gr_d(\HS_A\otimes \HS_B) \rightarrow
               \gr_{d^k}(\HS_{A^{\otimes k}}\otimes\HS_{B^{\otimes k}})\\
              &\text{ which maps } S\longmapsto S^{\otimes k},
    \end{aligned}\\
    &\begin{aligned}
      \phi_2: &\gr_{d}(\HS_{A^{\otimes k}}\otimes
                 \HS_{B^{\otimes k}}) \rightarrow
               \gr_{d_A^kd_B^k-d}(\HS_{A^{\otimes k}}\otimes
                 \HS_{B^{\otimes k}})\\
              &\text{ which maps } S\longmapsto S^{\perp}.
    \end{aligned}\raisetag{1em}
  \end{align}
  \end{subequations}
  We then have
  \begin{equation}
    \begin{split}
      &E_d^k(\HS_A,\HS_B)\\
      &\mspace{10mu}=\{
          S\in \gr_d(\HS_A\otimes\HS_B)
          \,|\,\phi_2\circ\phi_1(S)\cap \Sigma_{d_A^k-1, d_B^k-1}
          \neq\emptyset
        \}
    \end{split}
  \end{equation}
  where $\Sigma_{d_A^k-1, d_B^k-1}$ is the Segre variety (the
  projective variety consisting of all product states). If we let
  \begin{equation}
    \begin{split}
      T = \{&S\in \gr_{d_A^kd_B^k-d^k}(\HS_{A^{\otimes k}}\otimes
              \HS_{B^{\otimes k}})\,|\\
            &\mspace{150mu} S\cap\Sigma_{d_A^k-1, d_B^k-1}\neq\emptyset
          \}
    \end{split}\raisetag{2.7em}
  \end{equation}
  then $E_d^k(\HS_A,\HS_B)=(\phi_2\circ\phi_1)^{-1}(T)$. $\phi_1$ and
  $\phi_2$ are both proper morphisms, thus their composition is again a
  proper morphism, which implies that the pre-image
  $E_d^k(\HS_A,\HS_B)=(\phi_2\circ\phi_1)^{-1}(T)$ is Zariski closed if
  $T$ is Zariski closed.

  In the next step, we will prove the general result that
  \begin{equation}
    \begin{split}
      &R_{d}(\HS_A,\HS_B)\\
      &\quad=\{
          S\in \gr_{d}(\HS_{A}\otimes \HS_{B})
          \,|\, S \cap \Sigma_{d_A-1,d_B-1} \neq\emptyset
        \}
    \end{split}
  \end{equation}
  is Zariski closed, which will imply $T$ is Zariski closed. Let
  \begin{equation}
    X=\{(\mathcal{S},[v]) \,|\,
        \mathcal{S}\subset\HS_A\otimes\HS_B,
        [v]\in\Sigma_{d_A-1,d_B-1}
        \textrm{ and } v \in\mathcal{S}\}.
  \end{equation}
  Then $X$ is a subset of $\gr_{d}(\HS_{A}\otimes\HS_{B})\times
  \Sigma_{d_A-1, d_B-1}$. Let $P$ be the projection from
  $\gr_{d}(\HS_{A}\otimes\HS_{B})\times \Sigma_{d_A-1, d_B-1}$ to
  $\gr_{d}(\HS_{A}\otimes\HS_{B})$, so that that $R_{d}(\HS_A,
  \HS_B)=P(X)$. It is not hard to check that $X$ is Zariski closed.
  Since $\Sigma_{d_A-1, d_B-1}$ is a projective variety it is complete,
  and as a result the image of projection $P$ on any Zariski-closed set
  in $\gr_{d}(\HS_{A}\otimes\HS_{B})\times \Sigma_{d_A-1, d_B-1}$ is
  again Zariski closed. Therefore $R_{d}(\HS_A,\HS_B)=P(X)$ is Zariski
  closed.
\end{IEEEproof}

We will consider the case when $\HS_A = \HS_B =\CC^{d_A}$. In what
follows, it will be useful to represent $\HS_A\otimes\HS_B = \CC^{d_A}
\otimes \CC^{d_A}$ as the real vector space $\RR^2 \otimes \RR^{d_A}
\otimes \RR^{d_A}$. The complex Grassmannian $\gr_d(\CC^{d_A} \otimes
\CC^{d_A})$ can then be mapped injectively to the real Grassmannian
$\gr_{2d}(\RR^2\otimes \RR^{d_A} \otimes \RR^{d_A})$. Define $i$ to be a
linear operator acting on $\RR^2\otimes \RR^{d_A} \otimes \RR^{d_A}$ in
the natural way, i.e.\ as $\left(\begin{smallmatrix} 0 & -1 \\ 1 & 0
  \end{smallmatrix}\right) \otimes \id_{d_A} \otimes \id_{d_A}$. Then
$S\in \gr_{2d}(\RR^2\otimes \RR^{d_A} \otimes \RR^{d_A})$ corresponds to
an element of $\gr_d(\CC^{d_A} \otimes \CC^{d_A})$ if and only if it
satisfies $iS=S$.

Now we use the fact that a Zariski-closed set in a complex
vector space is also Zariski-closed in the isomorphic real vector
space to obtain the following corollary to \cref{lem:Edk-closed}.
\begin{corollary}\label{lem:Ed_Zariski_closed}
  $E_{2d}^k(\HS_A,\HS_B)$ is Zariski-closed in the real Grassmannian
  $\gr_{2d}(\RR^2 \otimes \RR^{d_A}\otimes\RR^{2d_B})$.
\end{corollary}

We now consider the set of subspaces that satisfy
\cref{eq:S_conjugate_symmetry2,eq:Sperp_conjugate_symmetry2} of
\cref{thm:S_strongly_unextendible_conditions}. Denote this set by
\begin{multline}
    F_d(\CC,d_A)
    =\{S\in\gr_d(\CC^{d_A} \otimes \CC^{d_A})\,|\\
       S = \flip(S),\;
       \flip(U\otimes V\cdot S) = U\otimes V\cdot S\,\}.
\end{multline}
To better handle the conjugate-linear constraints, we will consider
the equivalent set of real vector spaces, defined to be
\begin{equation}\label{eq:Fd-real}
  \begin{split}
    &F_d(\RR,d_A)\\
    &\mspace{20mu}=\{
       S\in\gr_{2d}(\RR^2\otimes \RR^{d_A} \otimes \RR^{d_A})
         \qquad\\
       &\mspace{60mu}
         S = iS,\; S = \flip(S),\;
         \flip(U\otimes V\cdot S) = U\otimes V\cdot S
     \}.
  \end{split}\raisetag{3.5em}
\end{equation}
While $F_d(\RR,d_A)$ and $F_d(\CC,d_A)$ are isomorphic, we will find it
convenient to work with both of them at different times.

As the following lemma shows, this set is also algebraic:
\begin{lemma}\label{lem:Fd_Zariski_closed}
  $F_d(\RR,d_A)$ is Zariski-closed in
  $\gr_{2d}(\RR^2\otimes\RR^{d_A}\otimes\RR^{d_A})$.
\end{lemma}
\begin{IEEEproof}
  We will prove a more general statement. If $\HS$ is a
  finite-dimensional real vector space, and $M\in\mathcal{B}(\HS)$ then
  define the action of $M$ on $\gr_{d}(\HS)$ by
  \begin{equation}
    M(S) = \{M\ket{\psi} : \ket{\psi}\in S\},
  \end{equation}
  for $S\in \gr_d(\HS)$. Then we claim that the set of subspaces
  invariant under $M$, $\{S\in\gr_d(\HS) : M(S)=S\}$, is Zariski-closed
  in $\gr_d(\HS)$.

  To show the lemma follows from this claim, take $\HS=\RR^2\otimes
  \RR^{d_A}\otimes \RR^{d_A}$ and $M$ to be in turn $i$, $\flip$, and
  $(U\otimes V)\cdot \flip$. Then use the fact that the intersection of
  two Zariski-closed sets is also Zariski-closed.

  To prove our claim about $\{S\in\gr_d(\HS) : M(S)=S\}$, we will use the
  Pl\"ucker embedding~\cite{Harris}. The Pl\"ucker embedding $\iota$ is a
  map from $\gr_d(\HS)$ into $\PP(\wedge^d \HS)$. Here $\wedge^d \HS$
  denotes the $d^\mathrm{th}$ exterior power of $\HS$, and $\PP$
  indicates that we are taking the projectification of $\wedge^d \HS$. If
  $S$ is spanned by $\{\ket{\psi_1},\ldots,\ket{\psi_d}\}$ then
  $\iota(S)$ is defined to be $\ket{\psi_1}\wedge \ldots \wedge
  \ket{\psi_d}$. To see that $\iota$ is a well-defined map, observe that
  replacing $\ket{\psi_i}$ by $\sum_{j=1}^d A_{i,j} \ket{\psi_j}$ for an
  invertible matrix $A$ has the effect of replacing $\iota(S)$ by
  $\det(A)\iota(S)$, which in projective space makes no difference.

  The exterior product $\ket{\psi_1}\wedge \ldots \wedge \ket{\psi_d}$
  can also be written as
  \begin{equation}
    \sum_{\sigma\in\mathcal{S}_d} (-1)^{\sgn(\sigma)}
    \ket{\psi_{\sigma(1)}} \otimes \cdots \otimes
    \ket{\psi_{\sigma(d)}}
  \end{equation}
  where $\mathcal{S}_d$ is the symmetric group on $d$ elements and
  $\sgn(\sigma)$ is the sign of the permutation $\sigma$. In this picture
  we have
  \begin{equation}
    \iota\bigl(M(S)\bigr) = M^{\otimes d}\;\iota(S).
  \end{equation}
  Thus the condition that $M(S)=S$ is equivalent to demanding that
  $\iota(S) = M^{\otimes d}\;\iota(S)$. This is a linear constraint on
  $\iota(S)$, so $\{\iota(S) : \iota(S) = M^{\otimes d}\;\iota(S)\} =
  \{\iota(S) : M(S)=S\}$ is Zariski-closed. But $\iota$ is a proper
  morphism, 
  so $\{S : M(S)=S\}$ must also be Zariski-closed, which completes the
  proof.
\end{IEEEproof}

\vspace{0.5em}
The following follows immediately from
\cref{lem:Ed_Zariski_closed,lem:Fd_Zariski_closed}:
\begin{corollary}
  $E_d^k(\HS_A,\HS_{A'})\cap F_d(\RR,d_A)$ is Zariski-closed in
  $F_d(\RR,d_A)$.
\end{corollary}

Any Zariski-closed subset has zero measure (in the usual Haar measure),
unless it is the full space. Thus $\bigcup_kE_d^k$, which is a countable
union of zero-measure sets, is either zero-measure or it is the full
space. Conversely, since $U_d$ is the complement of this union, it is
either full measure or it is the empty set. Since the intersection of two
Zariski-closed sets is Zariski-closed, the identical argument also holds
for $E_d^k\cap F_d$ and $U_d\cap F_d$, hence:
\begin{theorem}\label{thm:Ud_full_or_zero_measure}
  If the set $U_d(\HS_A,\HS_{A'})\cap F_d(\RR,d_A) \neq \emptyset$, then
  it is full measure in $F_d(\RR,d_A)$.
\end{theorem}

\subsection{Existence of a strongly unextendible
  con\-ju\-gate-sym\-metric subspace}
We now proceed to show that $U_d(\HS_A,\HS_{A'})\cap F_d(\RR,d_A)$ is not
empty. We will do this by starting with a family of strongly unextendible
subspaces and symmetrising them, so we need to get a handle on how much
the symmetrisation blows up the dimension of the subspace, which is the
content of the following lemma.
\begin{lemma}\label{lem:symmetrisation_blow-up}
  Let $\mathcal{F}:\gr_d(\HS_A\otimes\HS_{A'}) \to
  \bigcup_{d'}F_{d'}(\CC,d_A)$ be the map that symmetrises a subspace $S$
  by alternately iterating the maps $\mathcal{F}_1(S) = S + \flip(S)$ and
  $\mathcal{F}_2(S) = S + \flip_{U\otimes V}(S) = S + U^\dagger\otimes
  V^\dagger\; \flip(U\otimes V\cdot S)$ until convergence. Then, for
  \begin{equation}
    U = \id,\quad
    V = \begin{pmatrix}&&&&1\\ &&&1\\ &&\iddots \\ &1\\1\end{pmatrix}
     := X,
  \end{equation}
  the dimension $d'$ of the image $\mathcal{F}(S)$ satisfies $d'\leq 4d$.
\end{lemma}
\begin{IEEEproof}
  Let $M$ be an element of $\mat(S)$, and consider the action of $\flip$
  and $\flip_{U\otimes V}$ on $\mat(S)$. Since $X^\dagger = X$ and
  $X^2=\id$, we have
  \begin{gather}
    \flip(M) = M^\dagger,\\
    \flip_{U\otimes V}(M) = XM^\dagger X,\\
    \flip\circ\flip(M)
      = \flip_{U\otimes V}\circ\flip_{U\otimes V}(M) = M,\\
    \flip\circ\flip_{U\otimes V}(M)
      = \flip_{U\otimes V}\circ\flip(M) = XMX.
  \end{gather}
  Thus the alternating application of $\mathcal{F}_1$ and $\mathcal{F}_2$
  converges after a finite number of iterations, and maps a basis
  $\{M_i\}$ for $\mat(S)$ to a basis
  $\{M_i,M_i^\dagger,XM_iX,XM_i^\dagger X\}$ for
  $\mat(\mathcal{F}(S))$. The dimension of $S$ therefore increases by at
  most a factor of four (with equality when
  $\{M_i,M_i^\dagger,XM_iX,XM_i^\dagger X\}$ are all linearly
  independent).
\end{IEEEproof}

The other ingredient, namely a family of strongly unextendible subspaces,
is provided by the well-known \keyword{unextendible product bases}.
\begin{definition}
  An \keyword{unextendible product basis} (UPB) is a set of product
  states $\{\ket[AB]{\psi_i}\}$ (not necessarily orthogonal) in a
  bipartite space $\HS_A\otimes\HS_B$ such that
  $(\vspan\{\ket{\psi_i}\})^\perp$ contains no product states. The
  \keyword{dimension} of a UPB is the number of product states in the
  set.
\end{definition}
Clearly, a UPB spans a $1$-unextendible subspace. That this subspace is
in fact strongly unextendible is shown by the following lemma.
\begin{lemma}\label{lem:UPB_strong_unextendible}
  If $\{\ket[A_1B_1]{\psi_i^1}\}$ and $\{\ket[A_2B_2]{\psi_i^2}\}$ are
  unextendible product bases in $\HS_{A_1}\otimes\HS_{B_1}$ and
  $\HS_{A_2}\otimes\HS_{B_2}$ respectively, then
  $\{\ket{\psi_i^1}\ket{\psi_j^2}\}_{i,j}$ is an unextendible product
  basis in $\HS_{A_1A_2}\otimes\HS_{B_1B_2}$.
\end{lemma}
\begin{IEEEproof}
  If $\{\ket[A_1B_1]{\psi_i^1}\}$ and $\{\ket[A_2B_2]{\psi_i^2}\}$ are
  both orthogonal unextendible product bases, this case was proved in
  Ref.~\cite{upb03}. For non-orthogonal unextendible product bases, let
  $\ket[A_1B_1]{\psi_i^1}=\ket[A_1]{\alpha^1_i}\ket[B_1]{\beta^1_i}$ for
  $1\leq i\leq k_1$ and
  $\ket[A_2B_2]{\psi_j^2}=\ket[A_2]{\alpha^2_j}\ket[B_2]{\beta^2_j}$ for
  $1\leq j\leq k_2$.

  Assume for contradiction that $\{\ket{\psi_i^1}\ket{\psi_j^2}\}_{i,j}$
  is extendible in $\HS_{A_1A_2}\otimes\HS_{B_1B_2}$ which means there
  exists a product state $\ket[A_1A_2]{x}\ket[B_1B_2]{y}$ in
  $\HS_{A_1A_2}\otimes\HS_{B_1B_2}$ which is orthogonal to any
  $\ket[A_1B_1]{\psi_i^1}\ket[A_2B_2]{\psi_j^2}$. We then have
  \begin{equation}
    \forall\; 1\leq i\leq k_1 \text{ and } 1\leq j\leq k_2:\;
    \braket{\alpha^1_i,\alpha^2_j}{x}
    \braket{\beta^1_i,\beta^2_j}{y} = 0.
  \end{equation}

  For an $m\times n$ matrix $A$, $\vec(A)$ is an $mn$--element column
  vector whose first $m$ elements are the first column of $A$, the next
  $m$ elements are the second column of $A$, and so on. Thus ``$\vec$''
  converts the matrix into a vector. In ``$\vec$'' notation, we have
  $\vec(ABC)=(C^{T}\otimes A)\vec(B)$. Applying this, we obtain
  \begin{subequations}
  \begin{align}
    0&=\braket[A_1A_2]{\alpha^1_i,\alpha^2_j}{x}
       \braket[B_1B_2]{\bar{y}}{\bar{\beta}^1_i,\bar{\beta}^2_j}\\
     &\begin{aligned}
       =&\bra[A_1]{\alpha_i^1}
         \bigl(\id_{A_1}\otimes\bra[A_2]{\alpha^2_j}\bigr)\\
        &\mspace{30mu}
         \bigl(\ket[A_1A_2]{x}\bra[B_1B_2]{\bar{y}}\bigr)
         \bigl(\id_{B_1}\otimes\ket[B_2]{\bar{\beta}^2_j}\bigr)
         \ket[B_1]{\bar{\beta}^1_i}
     \end{aligned}\\[0.5em]
     &\begin{aligned}
       =&\bigl(
         \ket[B_1]{\bar{\beta}^1_i}^T\otimes\bra[A_1]{\alpha^1_i}
        \bigr)\cdot\\
       &\mspace{30mu}
        \vec\Bigl[
          \bigl(\id_{A_1}\otimes\bra[A_2]{\alpha^2_j}\bigr)\cdot\\
          &\mspace{100mu}
          \bigl(\ket[A_1A_2]{x}\bra[B_1B_2]{\bar{y}}\bigr)\cdot
          \bigl(\id_{B_1}\otimes\ket[B_2]{\bar{\beta}^2_j}\bigr)
       \Bigr]
     \end{aligned}\raisetag{3.5em}\\[0.5em]
    &\begin{aligned}
      =&\bra[B_1]{\beta^1_i}\bra[A_1]{\alpha^1_i}
       \Bigl(
         (\id_{B_1}\otimes\ket[B_2]{\bar{\beta}^2_j})^T\otimes
         (\id_{A_1}\otimes \bra[A_2]{\alpha^2_j})
       \Bigr)\cdot\\
      &\mspace{30mu}
       \vec\Bigl(\ket[A_1A_2]{x}\bra[B_1B_2]{\bar{y}}\Bigr)
    \end{aligned}\raisetag{1.5em}\\[0.5em]
    &\begin{aligned}
      =&\bra[B_1]{\beta^1_i}\bra[A_1]{\alpha^1_i}
        \left(
          \id_{B_1}\otimes\bra[B_2]{\beta^2_j}\otimes
          \id_{A_1}\otimes\bra[A_2]{\alpha^2_j}
        \right)\cdot\\
       &\mspace{30mu}
        \left(\bra[B_1B_2]{\bar{y}}^T\otimes\ket[A_1A_2]{x}\right)
        \vec(\id)
    \end{aligned}\\[0.5em]
    &\begin{aligned}
      =&\bra[B_1]{\beta^1_i}\bra[A_1]{\alpha^1_i}
        \Bigl[
          \Bigl(\id_{B_1}\otimes\bra[B_2]{\beta^2_j}\Bigr) \ket[B_1B_2]{y}
          \otimes\\
          &\mspace{170mu}
          \Bigl(\id_{A_1}\otimes\bra[A_2]{\alpha^2_j}\Bigr) \ket[A_1A_2]{x}
        \Bigr].
    \end{aligned}
  \end{align}
  \end{subequations}

  For any fixed $j$, the term in square brackets is either the zero
  vector, or a product state in $\HS_{A_1}\otimes\HS_{B_1}$ which is
  orthogonal to any
  $\ket[A_1B_1]{\psi_i^1}=\ket[A_1]{\alpha^1_i}\ket[B_1]{\beta^1_i}$.
  But $\{\ket[A_1B_1]{\psi_i^1}\}$ is an unextendible product basis, so
  if it is non-zero for some $j$ then we have a contradiction.

  Otherwise, $\bigl(\id\otimes\bra{\beta^2_j}\bigr) \ket[B_1B_2]{y}
  \otimes \bigl(\id\otimes\bra{\alpha^2_j}\bigr) \ket[A_1A_2]{x} = 0$ for
  any $j$. Let $\ket[B_1]{\gamma}$ and $\ket[A_1]{\delta}$ be two vectors
  such that $\bigl(\bra[B_1]{\gamma}\otimes
  \id_{B_2}\bigr)\ket[B_1B_2]{y}\neq 0$ and
  $\bigl(\bra[A_1]{\delta}\otimes \id_{A_2}\bigr)\ket[A_1A_2]{x}\neq
  0$. Then we have
  \begin{multline}
    \bra[B_2]{\beta^2_j}\bra[A_2]{\alpha^2_j}
    \Bigl[
      \bigl(\bra[B_1]{\gamma}\otimes \id_{B_2}\bigr) \ket[B_1B_2]{y}
      \otimes\\
      \bigl(\bra[A_1]{\delta}\otimes \id_{A_2}\bigr) \ket[A_1A_2]{x}
    \Bigr] = 0
  \end{multline}
  for any $j$. Here, the term in square brackets is a nonzero product
  state in $\HS_{A_2}\otimes\HS_{B_2}$ which is orthogonal to any
  $\ket[A_2B_2]{\psi_j^2}=\ket[A_2]{\alpha^2_j}\ket[B_2]{\beta^2_j}$. But
  $\{\ket[A_2B_2]{\psi_j^2}\}$ is also an unextendible product basis,
  which gives a contradiction as before.
\end{IEEEproof}

\Cref{lem:UPB_strong_unextendible} says that tensor products of
unextendible product bases are unextendible, which in particular implies
that all tensor powers of an unextendible product basis are unextendible,
i.e.\ unextendible product bases span strongly unextendible
subspaces. The following lemma giving the minimal dimension of a UPB was
proven in Ref.~\cite{Bhat04}:
\begin{lemma}\label{lem:UPB_strongly_unextendible}
  There exists a UPB of dimension $m$ in $\CC^{d_A}\otimes\CC^{d_B}$ for
  any $d_A + d_B - 1 \leq m \leq d_Ad_B$.
\end{lemma}

We are now in a position to prove the existence of strongly unextendible
subspaces in $F_d$ (i.e.\ strongly unextendible subspaces obeying the
symmetry constraints of
\cref{eq:S_conjugate_symmetry2,eq:Sperp_conjugate_symmetry2} from
\cref{thm:S_strongly_unextendible_conditions}), for sufficiently large
dimension. (It turns out that 16 is ``sufficiently large'' enough.)

\begin{lemma}\label{lem:strongly_unextendible_existence}
  For $U = \id$, $V = X$, there exist strongly unextendible subspaces
  $S\in F_d(\CC, d_A)$ of dimension $d$ for any $4(2d_A-1) \leq d \leq
  d_A^2$.
\end{lemma}
\begin{IEEEproof}
  Let $S$ be a subspace spanned by a UPB with the minimal dimension $m =
  2d_A - 1$. \Cref{lem:UPB_strongly_unextendible} tells us that $S$ is
  strongly unextendible. By \cref{lem:symmetrisation_blow-up}, its
  symmetrisation $\mathcal{F}(S)$ has dimension at most $4m =
  4(2d_A-1)$. Also, since symmetrising can never shrink the subspace, we
  have $\mathcal{F}(S)^\perp\subseteq S^\perp$ so $\mathcal{F}(S)$ is
  also strongly unextendible.

  Thus $\mathcal{F}(S)$ is a strongly unextendible subspace of dimension
  at most $4(2d_A-1)$. The lemma follows from the fact that any extension
  $S'\supseteq S$ is strongly unextendible if $S$ is.
\end{IEEEproof}

Combining
\cref{thm:Ud_full_or_zero_measure,lem:strongly_unextendible_existence},
we have shown that:
\begin{corollary}\label{cor:Ud_full_measure}
  For $d \geq 4(2d_A-1)$ and $U = \id$, $V = X$, the set of strongly
  unextendible subspaces $U_d(\HS_A,\HS_{A'})\cap F_d(\CC, d_A)$ is full
  measure in $F_d(\CC, d_A)$.
\end{corollary}
This leads to the main theorem of this section.
\begin{theorem}\label{thm:conjugate_symmetric_strongly_unextendible}
  For $d_A \geq 16$, $U = \id$, $V = X$, and for a subspace
  $S\in\CC^{d_A}\otimes\CC^{d_A}$ of dimension $4(2d_A-1) \leq d \leq
  d_A^2 - 4(2d_A-1)$ chosen uniformly at random subject to the symmetry
  constraints $\flip(S)=S$ and $\flip(U\otimes V\cdot S) = U\otimes
  V\cdot S$, both $S$ and $S^\perp$ will almost-surely be strongly
  unextendible.
\end{theorem}
\begin{IEEEproof}
  \Cref{cor:Ud_full_measure} implies that $S$ chosen in this way will
  almost-surely be strongly unextendible. But $S^\perp$ is then a random
  subspace subject to the same symmetry constraints, with dimension
  $4(2d_A-1) \leq d^\perp = d_A^2 - d \leq d_A^2 - 4(2d_A-1)$. Thus
  \cref{cor:Ud_full_measure} implies that $S^\perp$ will be almost-surely
  strongly unextendible. For there to exist a suitable $d$, we require
  $4(2d_A-1) \leq d_A^2 - 4(2d_A-1)$, or $d_A \geq 16$.
\end{IEEEproof}

\subsection{Positive-semidefinite conjugate-symmetric subspaces}
\Cref{thm:conjugate_symmetric_strongly_unextendible} shows that a random
subspace satisfying the symmetry constraints of
\cref{eq:S_conjugate_symmetry2,eq:Sperp_conjugate_symmetry2} from
\cref{thm:S_strongly_unextendible_conditions} will in fact also
almost-surely satisfy the strong unextendibility requirements of
\cref{eq:S_strongly_unextendible,eq:Sperp_strongly_unextendible}. It
remains to show that the positive-semidefiniteness requirements of
\cref{eq:S_positivity2,eq:Sperp_positivity2} can also be satisfied
simultaneously.

\begin{theorem}\label{thm:conjugate_symmetric_positive-semidefinite}
  If $d_A$ is even, $\lfloor d/2 \rfloor \leq d_A^2/2 - 1$ and $U=\id$,
  $V=X$, then the set
  \begin{equation}
    \begin{split}
      P_d(d_A)
      =\{&S\in F_d(\CC, d_A)\,|\\
      & S\text{ and } U\otimes V\cdot S^\perp
             \text{ positive semidefinite}\}
    \end{split}\raisetag{2.5em}
  \end{equation}
  has non-zero measure in $F_d(\CC, d_A)$.
\end{theorem}

In order to prove
\cref{thm:conjugate_symmetric_positive-semidefinite}, we would like to
demonstrate a single $S\in F_d(\CC,d_A)$ that is strictly positive
definite, which would then imply that there is ball of nonzero measure
around it that is positive semidefinite.  This sort of argument can be
used in manifolds such as $\gr_d(\HS)$, but it is not clear that it
carries over to a more complicated set, such as $F_d(\CC,d_A)$.  Thus,
we will first need to determine the structure of $F_d(\CC,d_A)$ and
demonstrate that (when $d_A$ is even) it decomposes into a direct sum
of spaces which are simpler to analyze.  Later we will see that this
lets us apply the intuition from this paragraph to prove
\cref{thm:conjugate_symmetric_positive-semidefinite}.

\begin{lemma}\label{lem:Fd-structure}
  If $d_A$ is even and $U=\id$, $V=X$, then
  \begin{equation}
    F_d(\RR,d_A) \cong
    \bigsqcup_{k=0}^d \gr_k(\RR^{d_A^2/2}) \times \gr_{d-k}(\RR^{d_A^2/2}).
  \end{equation}
  The $\sqcup$ denotes disjoint union, meaning that an element of
  $F_d(\RR,d_A)$ can be uniquely identified by specifying an integer $0
  \leq k \leq d$ and elements of $\gr_k(\RR^{d_A^2/2})$ and
  $\gr_{d-k}(\RR^{d_A^2/2})$.
\end{lemma}

\begin{IEEEproof}[Proof of \cref{lem:Fd-structure}]
  Elements of $F_d(\RR,d_A)$ are $2d$-dimensional real subspaces of
  $\RR^2 \otimes \RR^{d_A} \otimes \RR^{d_A}$. As such, they can be
  expressed as rank-$2d$ projectors. The constraints in \cref{eq:Fd-real}
  defining $F_d(\RR,d_A)$ can be expressed as symmetries of these
  projectors. In particular, $\Pi\in F_d(\RR,d_A)$ if and only if $\Pi$
  is a rank-$2d$ projector satisfying $i\,\Pi\,i^T = \Pi$,
  $\flip\,\Pi\,\flip^T=\Pi$ and $(X \otimes X)\Pi (X\otimes X)= \Pi$.

  Initially we will consider the $i$ and $\flip$ symmetries. Let
  $\flip_\pm$ denote the $\pm 1$ eigenspaces of $\flip$. Since $\Pi$
  commutes with $\flip$, it must be the sum of a projector onto a
  subspace of $\flip_+$ and a projector onto a subspace of $\flip_-$. In
  other words, $\Pi = \Pi_+ + \Pi_-$ where $\Pi_{\pm}\,\flip =
  \flip\,\Pi_{\pm} = \pm \Pi_{\pm}$.  Since $i$ and $\flip$ anticommute,
  $i$ must map $\flip_\pm$ to $\flip_{\mp}$. Thus $i\,\Pi_+ i^T$ is a
  projector onto $\flip_-$ and $i\,\Pi_- i^T$ is a projector onto
  $\flip_+$. Combined with the fact that $i\,\Pi\,i^T = \Pi$ we obtain
  that $i\,\Pi_\pm i^T = \Pi_\mp$. We can thus assume that $\Pi = \Pi_+ +
  i\,\Pi_+ i^T$ where $\Pi_+$ is a projector onto $\flip_+$. Since $\Pi$
  has rank $2d$, $\Pi_+$ must have rank $d$.

  Since $X \otimes X $ commutes with $\flip$ and $\Pi$, we have that
  $\Pi_+$ must also commute with $X \otimes X$. This means we can write
  $\Pi_+$ as $\Pi_{++} + \Pi_{+-}$, where $\Pi_{++}$ is a projector onto
  a subspace of the $+1$ eigenspace of $X\otimes X$ and $\Pi_{+-}$
  projects onto a subspace of the $-1$ eigenspace of $X \otimes X$.

  Working backwards we can see that if $\Pi_{++}$, $\Pi_{+-}$ are
  arbitrary projectors with the appropriate supports and with ranks
  summing to $d$, then $\Pi = (\Pi_{++} + \Pi_{+-}) + i(\Pi_{++} +
  \Pi_{+-})i^T$ projects onto a subspace in $F_d(\RR,d_A)$. If $\Pi_{++}$
  has rank $k$ then our choice of $\Pi$ is equivalent to choosing an
  element of $\gr_k(\RR^{d_A^2/2}) \times \gr_{d-k}(\RR^{d_A^2/2})$.
\end{IEEEproof}

\begin{IEEEproof}[Proof of
  \cref{thm:conjugate_symmetric_positive-semidefinite}]
  To understand what it means to have non-zero measure in $F_d(\CC,d_A)$,
  we use \cref{lem:Fd-structure} and the fact that
  $\dim\gr_k(\RR^{d_A^2/2}) = (d_A^2/2 - k)k$. Thus
  \begin{align*}
    &\dim\left(\gr_k(\RR^{d_A^2/2}) \times \gr_{d-k}(\RR^{d_A^2/2})\right)\\
    &\quad= \left(\frac{d_A^2}{2} - k\right)k \left(\frac{d_A^2}{2}
       - d+k\right)(d-k)\\
    &\quad= k(d-k) \left(\frac{d_A^2}{2}\left(\frac{d_A^2}{2}-d\right)
       - k (d-k)\right),
  \end{align*}
  which takes its maximum value at $k=d/2$ (for $d$ even) or $k=(d\pm
  1)/2$ (for $d$ odd). This means that all but a measure-zero subset of
  $F_d(\CC,d_A)$ is contained in these values of $k$. Indeed, if $k$ is
  even then the component of $F_d(\CC,d_A)$ corresponding to
  $\gr_{d/2}(\RR^{d_A^2/2}) \times \gr_{d/2}(\RR^{d_A^2/2})$ has measure
  one in $F_d(\CC,d_A)$. If $k$ is odd then the components corresponding
  to $\gr_{(d+1)/2}(\RR^{d_A^2/2}) \times \gr_{(d-1)/2}(\RR^{d_A^2/2})$
  and $\gr_{(d-1)/2}(\RR^{d_A^2/2}) \times \gr_{(d+1)/2}(\RR^{d_A^2/2})$
  each have measure $1/2$. For the rest of the proof we will take $k$ to
  be $d/2$ for $d$ even or \mbox{$(d-1)/2$ for $d$} odd. Let
  $\hat{F}_d(\CC,d_A)$ denote the part of $\FF(\CC,d_A)$ corresponding to
  $\gr_{d/2}(\RR^{d_A^2/2}) \times \gr_{d/2}(\RR^{d_A^2/2})$ if $d$ is
  even or $\gr_{(d+1)/2}(\RR^{d_A^2/2}) \times
  \gr_{(d-1)/2}(\RR^{d_A^2/2})$ if $d$ is odd.

  In either case, it suffices to show that $P_d(d_A) \cap \hat
  F_d(\CC,d_A)$ has positive measure in $\hat F_d(\CC,d_A)$. To do so,
  we first construct a positive-\emph{definite} subspace $S\in \hat
  F_d(\CC, d_A)$, meaning a subspace $S$ with a positive-definite
  basis. We would also like $(\id\otimes X)\cdot S^\perp$ to be
  positive definite. Our intuition is that since the set of
  positive-definite matrices is open, finding one matrix implies the
  existence of an open set (with positive measure) of
  positive-definite matrices around it. To rigorously extend this
  intuition to positive-definite subspaces, we need to define a
  continuous map $\eta: \hat F_d(\CC,d_A) \mapsto \cB(\CC^{d_A})$
  satisfying:
\begin{itemize}
  \item for any $S'$, $\eta(S')\in S'$; and
  \item $\eta(S)$ is a positive-definite operator on $\CC^{d_A}$.
\end{itemize}
  These properties will guarantee that every $S'\in \hat
  F_d(\CC,d_A)$ that is sufficiently close to $S$ will belong to
  $P_d(d_A)\cap \hat F_d(\CC,d_A)$, implying that this set has
  non-zero measure and proving the theorem.

  We construct $\eta$ by letting $M_k\in S$ be a positive-definite
  matrix, and extending it to an orthonormal basis for $S$ denoted
  $\{M_1,\ldots,M_k\}$, such that $S=M_1\wedge M_2 \wedge \cdots
  \wedge M_k$. Then we define $\eta := i_{M_{k-1}}i_{M_{k-2}}\cdots
  i_{M_1}$, where $i$ denotes the interior product. The definition of
  interior product guarantees the that $\eta(S')\in S'$ for any
  subspace $S'$, and that $\eta(M_k)=M_k$, which we have assumed is
  positive definite. It follows that $\eta$ maps some neighborhood of
  $S$ to positive-definite matrices, and that this neighborhood is
  therefore a set of positive-definite subspaces with nonzero measure.

  It remains only to construct the desired $S$. As we have observed in
  \cref{prop:necessary}, for $S$ to be positive definite, it is
  sufficient for $\mat(S)$ to contain a single positive-definite element.
  In particular, we will choose $S$ to contain $ \ket{\omega} =
  \sum_{i=1}^{d_A}\ket{i,i}$. We will also require that $S$ be orthogonal
  to $(\id \otimes X)\ket{\omega}$ so that $(\id \otimes X)S^\perp$ also
  contains $\ket\omega$ and is positive definite. Note that this only
  works if $d_A$ is even, otherwise $\ket{\omega}$ and $(\id\otimes
  X)\ket{\omega}$ are not orthogonal.

  Both $\ket\omega$ and $(\id \otimes X)\ket\omega$ belong to the $+1$
  eigenspace of $X \otimes X$. Thus to choose $S$ we need only choose an
  additional $k-1$ dimensions for $\Pi_{++}$ (from a space of dimension
  $d_A^2/2-2$) as well as an arbitrary rank-$(d-k)$ projector $\Pi_{+-}$
  whose support is contained within the $-1$ eigenspace of $X\otimes X$
  (with dimension $d_A^2/2$). This is possible as long as $k\leq d_A^2/2
  - 1$ and $d-k \leq d_A^2/2$. Substituting our choice of $k$, we find
  that it suffices to take $\lfloor d/2 \rfloor \leq d_A^2/2 - 1$.
\end{IEEEproof}

\subsection{Superactivation of the zero-error capacity}
\label{sec:main-thm}
\Cref{thm:conjugate_symmetric_strongly_unextendible} shows that, for
suitable dimensions, a subspace chosen at random subject to the symmetry
constraints of
\cref{eq:S_conjugate_symmetry2,eq:Sperp_conjugate_symmetry2} from
\cref{thm:S_strongly_unextendible_conditions} will, with probability~1,
satisfy the strong unextendibility conditions of
\cref{eq:S_strongly_unextendible,eq:Sperp_strongly_unextendible}. But
\cref{thm:conjugate_symmetric_positive-semidefinite} shows that there is
a non-zero probability that such a random subspace will satisfy the
positivity conditions of \cref{eq:S_positivity2,eq:Sperp_positivity2}.
Therefore, there must exist at least one subspace $S$ satisfying all the
conditions of \cref{thm:S_strongly_unextendible_conditions}. Finally, we
use \cref{prop:sufficient} to translate $S$ and $U\otimes V\cdot S^\perp$
into channels and complete the proof of superactivation of the zero-error
classical capacity of quantum channels, as stated in \cref{thm:main}
(\cref{sec:intro}), the main result of this paper.

``Suitable dimensions'' are any set of channel input and output
dimensions $d_A$ and $d_B$, together with a number of Kraus operators
$d_E$, that simultaneously satisfy all the dimension requirements of
\cref{thm:conjugate_symmetric_strongly_unextendible,thm:conjugate_symmetric_positive-semidefinite}.
Note that, from \cref{prop:sufficient}, $d_E$ is given by the dimension
of the subspace. In fact, the upper bound on the subspace dimension from
\cref{thm:conjugate_symmetric_positive-semidefinite} is always satisfied
if that of \cref{thm:conjugate_symmetric_strongly_unextendible} is. Also,
the requirement from \cref{thm:conjugate_symmetric_positive-semidefinite}
that $d_A$ be even merely implies that the input dimension to the channel
itself must be \emph{larger} than an even number, since we can always
embed a channel in a higher-dimensional input space. So the minimal
dimension requirements reduce to those stated in \cref{thm:main}.

\section{Conclusions}\label{sec:conclusions}
Smith and Yard's result~\cite{graeme+jon} showed that the capacity of
quantum channels to communicate quantum information behaves in the most
surprising way conceivable: two channels with zero capacity for
transmitting quantum information can nonetheless transmit quantum
information when used together (\keyword{superactivation}). On the other
hand, although it may well be non-additive~\cite{Hastings}, the usual
classical Shannon capacity of quantum channels cannot behave in this
extreme way.

However, in this work we have shown that the capacity of a quantum
channel for transmitting classical information \emph{perfectly}, the
zero-error classical capacity, exhibits the same surprising phenomenon as
the quantum capacity: two channels with zero capacity for perfect
transmission of classical information can nonetheless transmit classical
information perfectly when used together. This is, to our knowledge, the
first ever proven superactivation of a \emph{classical} capacity of a
standard quantum channel. (Note that although the zero-error capacity of
classical channels is non-additive, superactivation is impossible
classically.) It shows that this remarkable feature of quantum channels,
to allow communication when seemingly none should be possible, is not
restricted to quantum information but also occurs for classical
information.

How is this surprising behaviour possible? In the case of the quantum
capacity, superactivation is achieved \emph{without} the inputs to the
two channels needing to be entangled, and the intuition behind the
superactivation has more to do with local indistinguishability of
orthogonal quantum states~\cite{Oppenheim}. But entanglement \emph{is}
responsible for the superactivation of the zero-error capacity, just as
it is necessary if the standard classical Shannon capacity of quantum
channels is to be non-additive. So the fact that superactivation of the
zero-error classical capacity occurs for quantum but not for classical
channels can be attributed to the use of entangled inputs, which have no
classical analogue.

The results of \cref{sec:asymptotic} also resolve a number of other
questions. For one, they imply that the zero-error capacity of the
multi-sender/multi-receiver quantum channels of Duan and
Shi~\cite{Duan+Shi} can also be superactivated (extending their one-shot
result to the full asymptotic capacity). They also imply that even the
regularised version of the minimum output R\'enyi 0-entropy investigated
in Ref.~\cite{rank_additivity} is non-additive. In and of itself, this is
perhaps just a mathematical curiosity. But the same result for the
minimum output \emph{von Neumann} entropy (the R\'enyi 1-entropy) would
imply that the classical Shannon capacity of quantum channels really is
non-additive (i.e.\ that the capacity of two channels used together could
be greater than the sum of their individual capacities).

We close with an open question. Do there exist channels $\chan_1,
\chan_2$ with no zero-error classical capacity individually, but such
that $\chan_1 \otimes \chan_2$ has a positive zero-error \emph{quantum}
capacity?

\textbf{Note Added:} Simultaneously with our results, Duan~\cite{Duan}
extended his previous work to prove that the one-shot zero-error capacity
can also be superactivated in the case of single-input, single-output
channels. He also proves that the zero-error capacity is strongly
non-additive in the following sense: a quantum channel that has no
zero-error classical capacity can boost the zero-error capacity of a
second channel, which however does have some zero-error capacity on its
own. Whilst non-additivity of the zero-error capacity occurs even for
classical channels, this stronger form of non-additivity is impossible
classically. Both these results are implied by our stronger result, which
proves full superactivation in the standard sense (i.e.\ \emph{both}
channels have zero capacity) for the \emph{asymptotic} capacity (i.e.\
even infinitely many copies of the individual channels have zero
capacity). However, interestingly Duan's techniques are different to
ours, and also prove a similar non-additivity of the \emph{quantum} zero
error capacity, which our paper does not address.

\section*{Acknowledgements}
We would like to thank Andreas Winter for very useful discussions about
this work, and for pointing out the implications of our results to
non-additivity of regularised R\'enyi entropies. We also thank Runyao
Duan for kindly sending us a version of his latest results prior to
publication.

\enlargethispage{-5in}
\bibliographystyle{IEEEtran}
\bibliography{IEEEabrv,zero-error}

\begin{thebibliography}{10}
\providecommand{\url}[1]{#1}
\csname url@samestyle\endcsname
\providecommand{\newblock}{\relax}
\providecommand{\bibinfo}[2]{#2}
\providecommand{\BIBentrySTDinterwordspacing}{\spaceskip=0pt\relax}
\providecommand{\BIBentryALTinterwordstretchfactor}{4}
\providecommand{\BIBentryALTinterwordspacing}{\spaceskip=\fontdimen2\font plus
\BIBentryALTinterwordstretchfactor\fontdimen3\font minus
  \fontdimen4\font\relax}
\providecommand{\BIBforeignlanguage}[2]{{%
\expandafter\ifx\csname l@#1\endcsname\relax
\typeout{** WARNING: IEEEtran.bst: No hyphenation pattern has been}%
\typeout{** loaded for the language `#1'. Using the pattern for}%
\typeout{** the default language instead.}%
\else
\language=\csname l@#1\endcsname
\fi
#2}}
\providecommand{\BIBdecl}{\relax}
\BIBdecl

\bibitem{Holevo}
A.~S. Holevo, ``The capacity of the quantum channel with general signal
  states,'' \emph{IEEE Trans. Inform. Theory}, vol.~44, pp. 269--273, 1998,
  (arXiv:quant\nobreakdash-ph/9611023).

\bibitem{Schumacher+Westmoreland}
B.~Schumacher and M.~D. Westmoreland, ``Sending classical information via noisy
  quantum channels,'' \emph{Phys.\ Rev.\ A}, vol.~56, p. 131, 1997.

\bibitem{Devetak}
I.~Devetak, ``The private classical capacity and quantum capacity of a quantum
  channel,'' \emph{IEEE Trans. Inform. Theory}, vol.~51, p.~44, 2005,
  (arXiv:quant\nobreakdash-ph/0304127).

\bibitem{Shor}
P.~W. Shor, ``The quantum channel capacity and coherent information,'' MSRI
  seminar, November 2002.

\bibitem{Lloyd}
S.~Lloyd, ``Capacity of the noisy quantum channel,'' \emph{Phys.\ Rev.\ A},
  vol.~55, p. 1613, 1996.

\bibitem{DSS97}
D.~P. DiVincenzo, P.~W. Shor, and J.~A. Smolin, ``Quantum channel capacity of
  very noisy channels,'' \emph{Phys.\ Rev.\ A}, vol.~57, p. 830, 1998,
  (arXiv:quant\nobreakdash-ph/9706061).

\bibitem{Hastings}
M.~B. Hastings, ``A counterexample to additivity of minimum output entropy,''
  \emph{Nature Physics}, vol.~5, 2009, (arXiv:0809.3972
  [quant\nobreakdash-ph]).

\bibitem{Andreas+Patrick}
A.~J. Winter and P.~Hayden, ``Counterexamples to the maximal p-norm
  multiplicativity conjecture for all $p > 1$,'' \emph{Commun.\ Math.\ Phys.},
  vol. 284, no.~1, p. 263, 2008, (arXiv:0807.4753 [quant\nobreakdash-ph]).

\bibitem{rank_additivity}
T.~Cubitt, A.~W. Harrow, D.~Leung, A.~Montanaro, and A.~Winter,
  ``Counterexamples to additivity of minimum output p-r\'{e}nyi entropy for p
  close to 0,'' \emph{Commun.\ Math.\ Phys.}, vol. 284, p. 281, 2008,
  (arXiv:0712.3628 [quant\nobreakdash-ph]).

\bibitem{graeme+jon}
G.~Smith and J.~Yard, ``Quantum communication with zero-capacity channels,''
  \emph{Science}, vol. 321, p. 1812, 2008, (arXiv:0807.4935
  [quant\nobreakdash-ph]).

\bibitem{SST00}
P.~W. Shor, J.~A. Smolin, and A.~V. Thapliyal, ``Superactivation of bound
  entanglement,'' arXiv:quant\nobreakdash-ph/0005117, 2000.

\bibitem{private1}
K.~Li, A.~Winter, X.~Zou, and G.~Guo, ``Nonadditivity of the private classical
  capacity of a quantum channel,'' arXiv:0903.4308 [quant\nobreakdash-ph],
  2009.

\bibitem{private2}
G.~Smith and J.~Smolin, ``Extensive nonadditivity of privacy,'' arXiv:0904.4050
  [quant\nobreakdash-ph], 2009.

\bibitem{Shannon_zero-error}
C.~E. Shannon, ``The zero-error capacity of a noisy channel,'' \emph{IRE Trans.
  Inform. Theory}, vol. IT-2, p.~8, 1956.

\bibitem{zero-error_review}
J.~K{\"o}rner and A.~Orlitsky, ``Zero-error information theory,'' \emph{IEEE
  Trans. Inform. Theory}, vol.~44, no.~6, p. 2207, 1998.

\bibitem{MA05}
R.~A.~C. Medeiros and F.~M. de~Assis, ``Quantum zero-error capacity,''
  \emph{Int. J. Quant. Inf.}, vol.~3, p. 135, 2005.

\bibitem{BS07}
S.~Beigi and P.~W. Shor, ``On the complexity of computing zero-error and holevo
  capacity of quantum channels,'' arXiv:0709.2090 [quant\nobreakdash-ph], 2007.

\bibitem{Duan+Shi}
R.~Duan and Y.~Shi, ``Entanglement between two uses of a noisy multipartite
  quantum channel enables perfect transmission of classical information,''
  \emph{Phys.\ Rev.\ Lett.}, 2008, (arXiv:0712.3700 [quant\nobreakdash-ph]).

\bibitem{Jianxin_strong}
R.~Duan, J.~Chen, and Y.~Xin, ``Unambiguous and zero-error classical capacity
  of noisy quantum channels,'' (Manuscript in preparation).

\bibitem{Hartshorne77}
R.~Hartshorne, \emph{Algebraic Geometry}.\hskip 1em plus 0.5em minus
  0.4em\relax New York: Springer-Verlag, 1977.

\bibitem{Shafarevich94}
I.~R. Shafarevich, \emph{Basic algebraic geometry 1 (2nd, revised and expanded
  ed.)}.\hskip 1em plus 0.5em minus 0.4em\relax New York: Springer-Verlag,
  1994.

\bibitem{Schmidt_rank_subspace}
T.~Cubitt, A.~Montanaro, and A.~Winter, ``On the dimension of subspaces with
  bounded schmidt rank,'' \emph{J.\ Math. Phys.}, vol.~49, p. 022107, 2008,
  (arXiv:0706.0705 [quant\nobreakdash-ph]).

\bibitem{Harris}
J.~Harris, \emph{Algebraic Geometry}.\hskip 1em plus 0.5em minus 0.4em\relax
  Springer-Verlag, 1992.

\bibitem{upb03}
D.~DiVincenzo, T.~Mor, P.~W. Shor, J.~Smolin, and B.~Terhal, ``Unextendible
  product bases, uncompletable product bases and bound entanglement,''
  \emph{Commun.\ Math.\ Phys.}, vol. 238, no.~3, p. 379, 2003.

\bibitem{Bhat04}
R.~Bhat, ``A completely entangled subspace of maximal dimension,'' \emph{Int.
  J. Quant. Inf.}, vol.~4, no.~2, p. 325, 2006.

\bibitem{Oppenheim}
J.~Oppenheim, ``For quantum information, two wrongs can make a right,''
  \emph{Science Perspectives}, vol. 321, no. 5897, p. 1783, 2008.

\bibitem{Duan}
R.~Duan, ``Superactivation of zero-error capacity of noisy quantum channels,''
  arXiv:0906.2527 [quant\nobreakdash-ph], 2009.

\end{thebibliography}

\end{document}